\documentclass[11pt,oneside]{article}

\usepackage[english]{babel}

\usepackage[a4paper,total={165mm,240mm}]{geometry}
\usepackage{setspace}
\usepackage{mathtools} 
\usepackage{amsthm}
\usepackage{amssymb}
\usepackage{slashed}
\usepackage{float}

\usepackage{authblk}

\usepackage{graphicx}
\usepackage{tikz}
\usetikzlibrary{calc, decorations.pathmorphing, arrows.meta}
\usepackage{multirow}

\usepackage{pgfplots}
\pgfplotsset{compat=1.15}

\usepackage{xcolor}

\definecolor{sage}{RGB}{102,153,102}
\definecolor{olive}{RGB}{85,107,47}
\definecolor{emerald}{RGB}{46,184,114}
\colorlet{cblue}{cyan!50!white}
\colorlet{cgreen}{green!55!black}

\usepackage[autostyle]{csquotes}

\usepackage[style=alphabetic, backend=bibtex, useprefix]{biblatex}
\usepackage{hyperref}
\hypersetup{
    colorlinks=true,
    linkcolor=blue,
    citecolor=blue,
    urlcolor=blue,}

\theoremstyle{plain}

\newtheorem{thm}{Theorem}[section]
\newtheorem{prop}[thm]{Proposition}
\newtheorem{lem}[thm]{Lemma}
\newtheorem{cor}[thm]{Corollary}

\theoremstyle{definition}

\theoremstyle{remark}

\newtheorem{remark}{Remark}[section]

\newcommand{\R}{\mathbb{R}}

\newcommand{\de}{\partial}
\newcommand{\sfera}{\mathbb{S}}
\newcommand{\scri}{\mathcal{I}}
\newcommand{\Hcal}{\mathcal{H}}
\newcommand{\nablas}{\slashed{\nabla}}
\newcommand{\lra}{\longrightarrow}

\DeclarePairedDelimiter{\abs}{\lvert}{\rvert}

\title{Boundedness and decay for the conformal wave equation in Schwarzschild-AdS under dissipative boundary conditions}

\author{Alex Tullini\thanks{ \href{mailto:a.tullini@uni-muenster.de}{a.tullini@uni-muenster.de}}}

\affil{Universit\"at M\"unster, Mathematisches Institut, Einsteinstrasse 62, 48149 M\"unster, Bundesrepublik Deutschland}

\date{}

\begin{document}

\maketitle
\vspace{-2em}
\begin{abstract}
We study the conformal wave equation $\square_g \psi + \frac{2}{l^2} \psi = 0$ on 4-dimensional Schwarzschild–Anti de Sitter spacetimes under dissipative boundary conditions. We prove boundedness and decay of the non-degenerate energy of $\psi$ at an arbitrary polynomial rate of $(1+v)^{-n}$ provided that we control the (up to) $n$-times $T$-commuted energy. This contrasts with the inverse logarithmic decay obtained under Dirichlet boundary conditions and is in line with the result obtained in the pure Anti-de Sitter case under dissipative boundary conditions. In particular, the decay is not affected by the additional trapping at the photon sphere.
\end{abstract}

\setcounter{tocdepth}{3}
\tableofcontents

\section{Introduction}
This paper is motivated by the desire to understand the stability properties of black hole solutions to Einstein's vacuum equations:
\begin{equation} \label{EE} \tag{EE}
    Ric(g)-\dfrac{1}{2}R(g)g+\Lambda g=0.
\end{equation}
One can distinguish three fundamental classes of solutions to \eqref{EE}: asymptotically flat ($\Lambda=0$), asymptotically de Sitter ($\Lambda>0$), and asymptotically Anti–de Sitter ($\Lambda<0$), each containing a maximally symmetric representative — Minkowski, de Sitter, and Anti–de Sitter (AdS) spacetime, respectively (\cite{hawking2023large}).
We understand stability in a dynamical sense: given initial data for a specific solution of \eqref{EE}, we consider a small perturbation and study the evolution of perturbed initial data under \eqref{EE}. For $\Lambda\geq0$, this is possible because it was established in \cite{foures1952theoreme, choquet1969global} that \eqref{EE} decomposes into a set of constraint equations and evolution equations, which form a well-posed initial value problem. When $\Lambda<0$, while the decomposition of \eqref{EE} still holds, the construction of local in time solutions mimicking the asymptotic behaviour of AdS is only well-posed as an initial \emph{boundary} value problem, as established in \cite{friedrich1995einstein}. This is due to the presence of a conformal boundary at infinity, which is timelike in nature, where additional conditions need to be imposed.

Given the remarkable difficulty of understanding \eqref{EE} near a stationary solution, a common approach is to first look at the global in time dynamics of scalar perturbations on a fixed exterior, i.e. solutions to the conformal\footnote{Equation \eqref{confWE} is called conformal, or conformally invariant, because if we apply the conformal transformation $\tilde{g}=h^{2}g$, we find that $\tilde{\psi}=h^{-1}\psi$
solves equation $\square_{\tilde{g}}\tilde{\psi}-\frac{1}{6}\tilde{R}\tilde{\psi}=h^{-3}\left(\square_{g}\psi-\frac{1}{6}R\psi\right)$.} wave equation:
\begin{equation}\label{confWE}
    \square_g\psi -\dfrac{1}{6}R(g)\psi=0,
\end{equation}
where $\square_g=g^{\mu\nu}\nabla^2_{\mu\nu}$ is the trace of the Hessian and $g$ is the metric of the black hole exterior. This is a ``toy model" for stability because the decay properties of $\psi$ provide insight into the dispersive behaviour of linearized gravitational perturbations of $g$. Note that, if $g$ solves \eqref{EE}, we have that $R(g)=4\Lambda$ and the conformal wave equation becomes
\begin{equation}\label{confWE-AdS}
	\square_g\psi -\dfrac{2}{3}\Lambda\psi=0.
\end{equation} 
In this work, we study solutions to eq. \eqref{confWE-AdS} with $g$ the metric of Schwarzschild-AdS, subject to dissipative boundary conditions at infinity. To shed light on the effect of different boundary conditions, we briefly discuss the analogous problem in the simpler background of pure Anti-de Sitter space.
\subsection{The conformal wave equation in Anti-de Sitter space}\label{subsec:conf-WE-cylinder}
Anti-de Sitter space is the manifold $\mathcal{M}:=\R^4$ whose metric, in polar coordinates $(t,r,\theta,\varphi)\in\R\times(0,\infty)\times\sfera^2$, reads
\[ g_{\text{AdS}}=-\left(1+\frac{r^2}{l^2}\right)dt^2+\left(1+\frac{r^2}{l^2}\right)^{-1}dr^2+r^2(d\theta^2+\sin^2\theta d\varphi^2),\]
where $l^2=-\frac{3}{\Lambda}$ is the AdS radius. If we apply the coordinate transformation
\[ \tau=\frac{t}{l}\in\R, \quad 
\chi=\arctan\Big(\frac{r}{l}\Big)\in\Big(0,\frac{\pi}{2}\Big) \]
we find that
\[ g_{\text{AdS}}=\frac{l^2}{(\cos\chi)^2}(-d\tau^2+d\chi^2+\sin^2\chi(d\theta^2+\sin^2\theta d\varphi^2)).\]
Recall that the metric of the Einstein Cylinder $(\R\times\sfera^3,g_E)$, in spherical coordinates $(\tau,\chi,\theta,\varphi)$, is given by
\[g_{\text{E}}=-d\tau^2+d\chi^2+\sin^2\chi(d\theta^2+(\sin\theta)^2 d\varphi^2).\]
Therefore, $(\mathcal{M},g_{\text{AdS}})$ is conformally equivalent to \emph{half} of the Einstein Cylinder, that is $(\R\times\sfera^3_h,g_{\text{E}})$ with $\sfera^3_h$ the northern hemisphere of $\sfera^3$, since $\chi$ ranges between $0$ and $\frac{\pi}{2}$ instead of $(0,\pi)$.
Applying this conformal transformation (see Footnote 1) to equation \eqref{confWE}, we find that $v=\frac{\cos\chi}{l}\psi$ solves
\begin{equation}\label{confWE-cylinder}
    \square_{g_{\text{E}}}v-v=0
\end{equation}
on $(\R\times\sfera^3_h,g_{\text{E}})$. As a consequence, instead of studying the behaviour of solutions $\psi$ to \eqref{confWE-AdS} in AdS, we may study the behaviour of solutions $v$ to \eqref{confWE-cylinder} on $(\R\times\sfera^3_h,g_{\text{E}})$. We list the most common choices of boundary conditions for solutions $v$ in the cylindrical model and we translate them into conditions for $\psi$ at the conformal boundary of AdS.
\begin{table}[H]
\centering
{
\renewcommand{\arraystretch}{1.5}
\begin{tabular}{c|c|c}
 & (Einstein cylinder) $\chi \rightarrow \frac{\pi}{2}$ & (AdS) $r\rightarrow \infty$ \\ 
 \hline
 \textbf{Dirichlet} b.c. &  $ v \rightarrow 0$ & $r\psi \rightarrow 0$ \\
 \hline
 \textbf{Neumann} b.c. & $\de_{\chi}v\rightarrow 0$ & $\frac{r^2}{l^2}\de_r(r\psi) \rightarrow 0$ \\
 \hline
 \textbf{Robin} b.c. & $\de_{\chi}v + \beta v \rightarrow 0$ & $\frac{r^2}{l^2}\de_r(r\psi) + \beta\, (r\psi) \rightarrow 0$ \\
 \hline
 \textbf{Dissipative} b.c. & $\de_{\tau}v+\frac{1}{c}\de_{\chi}v\rightarrow 0$ & $\de_{t}(r\psi)+\frac{1}{c}\frac{r^2}{l^2}\de_r(r\psi)\rightarrow 0$\\
 \hline
 \textbf{Optimally Dissipative} b.c. & $\de_{\tau}v+\de_{\chi}v\rightarrow 0$ & $\de_{t}(r\psi)+\frac{r^2}{l^2}\de_r(r\psi)\rightarrow 0$
\end{tabular}}
\caption{Boundary conditions in Anti de Sitter space, where $\beta$ is a function on $\scri^+$ and $c\in(0,1]$.}
\label{table:boundary}
\end{table}
\noindent Notice how the first two can be viewed as special examples of Robin conditions, which we formally recover in the limit cases of $\beta=0$ and $\beta\rightarrow\infty$. They are typically referred to as \emph{reflective} boundary conditions.

The initial boundary value problem associated to \eqref{confWE-cylinder} is well-posed for any of the boundary conditions in the second column of Table \ref{table:boundary} -- provided that $\beta$ is sufficiently regular -- by standard techniques. Therefore, the initial boundary value problem associated to \eqref{confWE-AdS} is also well-posed with any of the boundary conditions in the third column --  provided, again, that $\beta$ is suitably regular. In fact, well-posedness holds for all asymptotically AdS spacetimes with any of the boundary conditions in Table \ref{table:boundary}; see \cite{holzegel2012well, vasy2012wave, warnick2013massive, holzegel2020asymptotic}.

\subsection{The exterior of Schwarzschild--Anti-de Sitter}
Let $M>0$ and $l>0$ be fixed parameters and let $r_+$ be the single real root of $(1-\frac{2M}{r}+\frac{r^2}{l^2})$. We describe the exterior $(\mathcal{R},g_{M,l})$ of Schwarzschild-AdS using ingoing coordinates $(v,r,\theta,\varphi)$, up to and including its future event horizon, as the Lorentzian manifold with boundary
\begin{align*}
    \mathcal{R}:= & \R\times[r_+,\infty)\times\sfera^2,  \\ 
    g_{M,l}:= & -\left(1-\frac{2M}{r}+\frac{r^2}{l^2}\right)dv^2+2dvdr+r^2(d\theta^2+(\sin\theta)^2d\varphi^2),
\end{align*}
We denote by $\Hcal^+$ the future event horizon $\Hcal^+:=\{(v,r,\theta,\varphi)\in\R\times[r_+,\infty)\times\sfera^2\ | \ r=r_+ \}$ and by $\scri^+$ the boundary at infinity.
\begin{figure}[H]
\centering
\scalebox{1.2}{
\begin{tikzpicture}
    \node (I) at (0,0) {};
    \node (II) at (2,-2) {\raisebox{-6mm}{$i^-$}};
    \node (III) at (2,2) {\raisebox{+6mm}{$i^+$}};
    \node (R) at (1.2,0.2) {{\small $\mathcal{R}$}};

    \fill (I) circle (0.5pt);
    \fill (II) circle (0.5pt);
    \fill (III) circle (0.5pt);
    
    \draw[dashed] (I.center)--(II.center) node[midway, below, xshift=-2mm] {$\Hcal^-$};
    \draw[dashed] (II.center)--(III.center) node[midway, xshift=4mm] {$\scri^+$}; 
    \draw (III.center)--(I.center) node[midway, above, xshift=-2mm] {$\Hcal^+$};
    
    \draw[thick] (0.2,0.2)--(2,-1.6) node[midway, above, xshift=1mm] {{\footnotesize $\Sigma_{v_0}$}};
\end{tikzpicture}}
\caption{Penrose diagram of the exterior $\mathcal{R}$ with a constant $v$ hypersurface $\Sigma_{v_0}$.}
\end{figure}
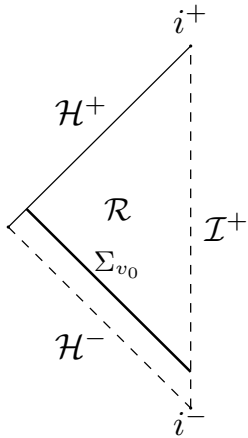
\noindent We denote the restriction of $g_{M,l}$ to the spheres of constant $(v,r)$ by $\slashed{g}= r^2(d\theta^2+(\sin\theta)^2d\varphi^2)$, use $\nablas$ for its Levi-Civita connection and $\slashed{\Delta}$ for its Laplace-Beltrami operator. We define the tortoise coordinate $r^*:(r_+,\infty)\lra (-\infty,\frac{\pi}{2})$ by requesting that
\[ \dfrac{dr^*}{dr}=\dfrac{1}{(1-\frac{2M}{r}+\frac{r^2}{l^2})} \quad \text{and} \quad  \lim\limits_{r\to\infty} r^*(r)=\frac{\pi}{2}.\]
We fix notation for the vector fields $T:=\de_v$ and $R^*:=\de_{r^*}$. Note that the vector field $\de_r|_{(v,r,\theta,\varphi)}$ is different from $\de_r|_{(t,r,\theta,\varphi)}$, with $(t,r,\theta,\varphi)$ the standard coordinates. In fact, we have
\[ R^*=T+\Big(1-\frac{2M}{r}+\frac{r^2}{l^2}\Big)\de_r\Big|_{(v,r,\theta,\varphi)}.\]
To avoid confusion, we introduce the notation
\[\de_{\tilde{r}}=\de_r|_{(v,r,\theta,\varphi)}.\]
We also introduce the vector fields $\Gamma_i$, for $i=1,2,3$, to be a basis of generators of the Lie Algebra of $SO(3)$. We will later use the notation $\Gamma^{\alpha}$, with $\alpha=(\alpha_1,\alpha_2,\alpha_3)$ being a multi-index, to denote $\Gamma^\alpha = \Gamma_1^{\alpha_1}\Gamma_2^{\alpha_2}\Gamma_3^{\alpha_3}$. 
\subsection{The conformal wave equation in Schwarzschild-AdS}
We study the initial boundary value problem associated to the conformal wave equation on $(\mathcal{R},g_{M,l})$,
\begin{equation} \label{WE} \tag{WE}
    \square_{g_{M,l}} \psi +\dfrac{2}{l^2}\psi=0,
\end{equation}
with characteristic initial data on the null hypersurface $\Sigma_{v_0}:=\{(v,r,\theta,\varphi)\in \mathcal{R} \ | \ v=v_0 \}$ and subject to (optimally) dissipative boundary conditions at $\scri^+$
\begin{equation}\label{BC} \tag{BC}
    2T(r\psi)+\frac{r^2}{l^2}\de_{\tilde{r}}(r\psi)\longrightarrow 0 \ \ \text{as} \ r\rightarrow\infty.
\end{equation}
\noindent This evolution problem is well-posed in light of the discussion in the previous section. In our present context, the well-posedness statement takes the following form -- see \cite{holzegel2020asymptotic}
 for the analogous statement in Anti-de Sitter.
\begin{thm}[Well-posedness]\label{thm:wp} Fix $\overline{v}>v_0$. Given smooth functions $\psi_0,\psi_1:\Sigma_{v_0}\rightarrow\R$ that satisfy suitable asymptotic conditions, there exists a unique smooth function $\psi$ solving \eqref{WE} in $\mathcal{R}\cap J^+(\Sigma_{v_0})$ such that
\begin{enumerate}
\item The initial conditions $\psi\Big|_{v_0}=\psi_0$ and $T\psi\Big|_{v_0}=\psi_1$ are satisfied;
\item The dissipative boundary conditions \eqref{BC} are satisfied;
\item We have the estimate \[\sup\limits_{[v_0,\overline{v}]\times [r_+,\infty)\times\sfera^2} \sum\limits_{m,n,k\leq K}\Big|T^m(r^2\de_{\tilde{r}})^n (r\nablas)^k (r\psi)\Big| \leq C(\overline{v}, \psi_0,\psi_1,K) \]
where $m,n,k$ are integers and the constant $C(\overline{v}, \psi_0,\psi_1,K)$ depends only on $\overline{v}, \psi_0,\psi_1,K$.
\end{enumerate} 
\end{thm}
\subsection{Main results} \label{subsec:mainresults}
We state the main results: an energy boundedness statement and an integrated decay estimate. We sketch their proofs in greater detail in Subsection \ref{subsec:sketch}. The notation $A\lesssim B$ stands for $A \leq C\cdot B$, where the constant $C>0$ only depends on $M$ and $l$. 

We first establish energy boundedness for the energy $\mathbb{E}^T[\psi]$, which is associated to the Killing field $T$ and, as such, degenerate at $\Hcal^+$. We then remove the degeneracy exploiting the redshift effect, with a technique well-known since  \cite{dafermos2008lectures, dafermos2009red}.
\begin{thm}[Energy boundedness]\label{thm:boundednessintro}
    Let $\psi$ be any smooth solution to \eqref{WE} subject to dissipative boundary conditions \eqref{BC}, arising from initial data as in Theorem \ref{thm:wp}. Let $v_0\leq v_1\leq v_2$. We have the following statements. 
    \begin{enumerate}
    \item The solution $\psi$ satisfies the (degenerate at $\Hcal^+$) energy conservation law
    \begin{equation*}
    \mathbb{E}^T[\psi](v_2)+\int\limits_{\scri^{+}(v_1,v_2)}(T(r\psi))^{2}dvd\omega+\int\limits_{\mathcal{H}^{+}(v_1,v_2)}(T(r\psi))^{2}dvd\omega=\mathbb{E}^T[\psi](v_1),
    \end{equation*}
    where the energy $\mathbb{E}^T[\psi]$ is given by
    \begin{equation*}
    \mathbb{E}^T[\psi](v):=\dfrac{1}{2}\int\limits_{\Sigma_{v}}\left[\Big(1-\frac{2M}{r}+\frac{r^2}{l^2}\Big)(\de_{\tilde{r}}(r\psi))^{2}+\abs{\nablas(r\psi)}^{2}+\dfrac{2M}{r^3}(r\psi)^{2}\right]drd\omega.
    \end{equation*}  
    In particular, we have the energy boundedness statement $\mathbb{E}^T[\psi](v_2)\leq \mathbb{E}^T[\psi](v_1)$.
    \item The solution $\psi$ satisfies the (non-degenerate) energy boundedness statement
    \begin{equation*}
        \overline{\mathbb{E}}[\psi](v_2)\lesssim \overline{\mathbb{E}}[\psi](v_1)
    \end{equation*}
    where the energy $\overline{\mathbb{E}}[\psi]$ is given by
    \begin{equation*}
    \overline{\mathbb{E}}[\psi](v):=\int\limits_{\Sigma_{v}}\left[r^2(\de_{\tilde{r}}(r\psi))^{2}+\abs{\nablas(r\psi)}^{2}+\dfrac{(r\psi)^{2}}{r^2}\right]drd\omega.
    \end{equation*}
    \end{enumerate}
\end{thm}
Then, we obtain an integrated decay estimate for the non-degenerate energy $\overline{\mathbb{E}}[\psi]$, for which we ultimately deduce (arbitrarily fast) inverse polynomial decay with derivative loss.
\begin{thm}[Integrated decay estimate] \label{thm:integrateddecayintro}
    Let $\psi$ be a smooth solution to \eqref{WE} subject to dissipative boundary conditions \eqref{BC}, arising from initial data as in Theorem \ref{thm:wp}. For any $v_0\leq v_1\leq v_2$, we have the following integrated decay estimate with derivative loss:
    \[  \mathbb{I}[\psi](v_1,v_2) \lesssim \overline{\mathbb{E}}[\psi](v_1)+\overline{\mathbb{E}}[T\psi](v_1). \]
    where
    \[\mathbb{I}[\psi](v_1,v_2)=\int\limits_{v_1}^{v_2}\int\limits_{r_+}^{\infty}\int\limits_{\sfera^2}\left[\dfrac{(T(r\psi))^2}{r^2}+r^2(\de_{\tilde{r}}(r\psi))^2+\abs{\nablas(r\psi)}^2+\dfrac{(r\psi)^2}{r^2}\right]dvdrd\omega.\]
\end{thm}
\begin{cor}[Arbitrarily fast polynomial decay]\label{cor:decay}
    Let $\psi$ be as in Theorem \ref{thm:integrateddecayintro}. Then, for any positive integer $n$ and for all $v_0\leq v$, the energy $\overline{\mathbb{E}}[\psi](v)$ satisfies the inequality
\[ \overline{\mathbb{E}}[\psi](v)\lesssim \dfrac{C_n}{(1+v)^n}\sum\limits_{i=0}^{n}\overline{\mathbb{E}}[T^i\psi](v_0),\]
where the constant $C_n$ depends on $n$.
\end{cor}
\subsection{Sketch of the proofs}\label{subsec:sketch}
In essence, both results follow from application of the vector field method over the region $D=[v_1,v_2]\times[r_+,\infty)\times\sfera^2$ in $(\mathcal{R},g_{M,l})$.
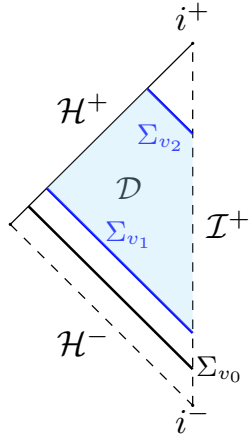
\begin{figure}[H] \label{fig:penrosediagramD}
\centering
\scalebox{1.2}{
\begin{tikzpicture}
    \node (I) at (0,0) {};
    \node (II) at (2,-2) {\raisebox{-6mm}{$i^-$}};
    \node (III) at (2,2) {\raisebox{+6mm}{$i^+$}};

    \fill (I) circle (0.5pt);
    \fill (II) circle (0.5pt);
    \fill (III) circle (0.5pt);
    
    \draw[dashed] (I.center)--(II.center) node[midway, below, xshift=-2mm] {$\Hcal^-$};
    \draw[dashed] (II.center)--(III.center) node[midway, xshift=4mm] {$\scri^+$}; 
    \draw (III.center)--(I.center) node[midway, above, xshift=-2mm] {$\Hcal^+$};

    \draw[thick] (0.2,0.2)--(2,-1.6) node[pos=1, xshift=3mm] {{\footnotesize $\Sigma_{v_0}$}};
    \draw[blue, thick] (0.4,0.4)--(2,-1.2) node[midway, above, xshift=1mm] {{\footnotesize $\Sigma_{v_1}$}};
    \draw[blue, thick] (1.5,1.5)--(2,1) node[midway, below, xshift=-1mm] {{\footnotesize $\Sigma_{v_2}$}};

    \node (IV) at (1.3,0.4) {{\small $\mathcal{D}$}};

    \fill[cyan!30, opacity=0.3] (0.4,0.4)--(1.5,1.5)--(2,1)--(2,-1.2)--cycle;
\end{tikzpicture}}
\caption{Penrose diagram of the exterior $\mathcal{R}$ with the (shaded) region $D$.}
\end{figure}
\textbf{Energy boundedness.} The natural energy $\int_{\Sigma_v}\mathbb{T}_{\mu\nu}T^{\nu}n_{\Sigma_v}dvol_{\Sigma_v}$ arising from the vector field current associated with the $T$ vector field\footnote{See Subsection \ref{subsubsec:vf-notation} for definitions of $\mathbb{T}_{\mu\nu}$ and of vector field currents.} is neither coercive nor finite, and thus requires normalisation. However, working with $r\psi$ instead of $\psi$ and using the following equivalent formulation of \eqref{WE}
\begin{equation}\label{WErpsi}
-TT(r\psi)+R^{*}R^{*}(r\psi)+\Big(1-\frac{2M}{r}+\frac{r^2}{l^2}\Big)\slashed{\Delta}(r\psi)-V(r)(r\psi)=0,
\end{equation}
with $V(r)=\frac{2M}{r^{3}}(1-2M/r+r^{2}/l^{2})$, the energy obtained multiplying by $T(r\psi)$ and integrating by parts is both positive and finite. This can be viewed as a simple case of twisted $r$-derivatives; see \cite{holzegel2014boundedness} for more general twisted objects in a similar setting. See Remark \ref{rem:JTcurrent} for further details.

\textbf{Integrated decay.} Consider the Morawetz multiplier 
\[X(r\psi)=f(r)R^*(r\psi)+\frac{f'(r)}{2}(r\psi) \quad \text{with} \quad f(r)=\Big(1-\frac{3M}{r}\Big)\sqrt{1+\frac{6M}{r}}, \]
where $f(r)$ is inspired by the one used in \cite{holzegel2024note} to obtain integrated decay in Schwarzschild--de Sitter. Multiply eq. \eqref{WErpsi} by $X(r\psi)$ and integrate by parts. The resulting bulk term is manifestly coercive, exception made for its zeroth order coefficient, which we control with additional work in the Appendix\footnote{This is easier in Schwarzschild--de Sitter because $r_+\geq 2M$ when $\Lambda\geq 0$, so positivity of the Schwarzschild limit can be used to prove positivity of the Schwarzschild-dS case. On the contrary, when $\Lambda\leq 0$, we have that $r_+\leq 2M$.}. The boundary terms are all easily controlled, exception made for the energy flux along $\scri^+$, which features a positive angular term (familiar from the same computation in AdS case, see \cite{holzegel2020asymptotic}) and a negative zeroth order term (absent in the AdS case, see again \cite{holzegel2020asymptotic}). Absorbing the zeroth order term into the angular term via a Poincaré-type inequality resolves the sign issue but introduces a restriction on $M,l$. This issue at $\scri^+$ appears to be inevitable when using the multiplier $X(r\psi)$. To avoid a conditional statement, we instead prove a (degenerate) Morawetz estimate for $T\psi$ and extend control to $r\psi$ and its spatial derivatives via an elliptic estimate, as detailed in the following two steps.
\begin{enumerate}
\item (Subsection \ref{subsec:degdecay}, Theorem \ref{thm:deg-integrated}) We consider eq. \eqref{WErpsi} commuted with $T\psi$, multiply by $XT(r\psi)$ and integrate over the region $[v_1,v_2]\times[r_+,\infty)\times\sfera^2$. In so doing, we obtain the (degenerate) integrated decay estimate for $T\psi$:
\[ \mathbb{I}_{\text{deg}}[T\psi](v_1,v_2)\lesssim \overline{\mathbb{E}}[\psi](v_1)+\overline{\mathbb{E}}[T\psi](v_1),\]
where
\[ \mathbb{I}_{\text{deg}}[T\psi](v_1,v_2):= \int\limits_{v_1}^{v_2}\int\limits_{r_+}^{\infty}\int\limits_{\sfera^2}\left[\dfrac{(R^{*}T(r\psi))^{2}}{r^3}+\dfrac{1}{r}\left(1-\dfrac{3M}{r}\right)^{2}\abs{\nablas T(r\psi)}^{2}+\dfrac{(T(r\psi))^{2}}{r^2}\right]dvdrd\omega.\]
The above estimate closes without restriction on the parameters because the negative zeroth order term along $\scri^+$ now features $(T(r\psi))^2$, which we control with $\mathbb{E}^T[\psi](v_1)$ in light of the energy conservation law (Theorem \ref{thm:boundednessintro}). Notice how this Morawetz estimate for $T\psi$ guarantees non-degenerate control of $(T(r\psi))^2/r^2$.
\item (Subsection \ref{subsec:elliptic}, proof of Theorem \ref{thm:integrateddecayintro}) We extend control from $(T(r\psi))^2$ to the spatial derivatives and to $r\psi$ itself using the structure of \eqref{WErpsi}, precisely by closing an elliptic estimate for the elliptic part of the wave operator. In so doing, we integrate by parts eq. \eqref{WErpsi} multiplied by $h(r)\cdot(r\psi)$ and find $h(r)(T(r\psi))^2/(1-\frac{2M}{r}+\frac{r^2}{l^2})$ in the bulk term. Ideally, one would choose $h\equiv 1$ but regularity forces $h(r)$ to vanish near $\Hcal^+$, thus yielding degenerate control near the horizon and an error due to the cut-off. We take care of both issues by invoking the redshift effect once more and are ultimately able to prove integrated decay for the non-degenerate energy $\overline{\mathbb{E}}$.
\end{enumerate}

Finally, using that $\overline{\mathbb{E}}$ is positive, bounded (Theorem \ref{thm:boundednessintro}), and satisfies integrated decay (Theorem \ref{thm:integrateddecayintro}), a pigeonhole argument yields Corollary \ref{cor:decay}.
\subsection{Further discussion}
Let us compare the main results of this paper with related results from the literature. In the simpler case of AdS, energy decay is largely determined by the amount of energy trapping at infinity, which we expect to be worse under reflective conditions and to improve under dissipative conditions. This is confirmed by rigorous results, since:
\begin{itemize}
\item under \textbf{Dirichlet} conditions, there are non-decaying, time-periodic, finite energy solutions to eq. \eqref{confWE-cylinder} like $v(\tau,\chi)=\sin(2\tau)\cos(\chi)$, which corresponds to the solution $\psi=\frac{\sin(2t/l)}{l(1+(r/l)^2)}$ of \eqref{confWE-AdS};
\item under \textbf{dissipative} conditions, we know from \cite{holzegel2020asymptotic} that the energy of solutions decays at an arbitrarily fast inverse polynomial rate in time, at the cost of controlling sufficiently many derivatives -- that is, we have energy decay with derivative loss.
\end{itemize}
The above results plus further analyses of gravitational perturbations presented in \cite{holzegel2020asymptotic}, lead the authors to conjecture that Anti-de Sitter spacetime is non-linearly unstable for reflective and asymptotically stable for optimally dissipative boundary conditions.
As regards asymptotically AdS black holes, we begin by mentioning \cite{holzegel2013decay, holzegel2014quasimodes}, where the authors study the problem in Kerr-AdS (under the Hawking-Reall bound on the parameters $(M,l,a)$), with Dirichlet boundary conditions, and establish an inverse logarithmic rate of decay. Although this decay rate suffices to establish linear stability \cite{graf2024linearI, graf2024linearII, graf2024linearIII}, it is expected to be too weak to control nonlinear effects.

For Kerr-AdS, the Hawking-Reall bound guarantees the absence of superradiant effects, thus making it possible to first obtain an energy boundedness statement and then establish decay. If the bound is violated, \cite{dold2017unstable} proves that, under Dirichlet conditions, there are exponentially growing mode solutions. See also \cite{graf2025stationary} for the construction of stationary 
mode solutions for linear gravitational perturbations, for parameters crossing the Hawking-Reall bound, 
and a complementary high-frequency mode stability result below the bound.
 
On the contrary, the case of Kerr-AdS with the Hawking-Reall bound violated and dissipative conditions at the boundary is still unexplored. If the Hawking-Reall bound is not violated by too much, one expects the redshift effect at the event horizon to be enough to mitigate superradiant effects. Then, based on the results of this paper, if dissipative conditions are imposed, one might expect energy to decay at an inverse polynomial rate. 
\begin{table}[H]
\centering
\small
{
\renewcommand{\arraystretch}{1.3}
\begin{tabular}{c|c|c|c}
 & \textbf{Anti-de Sitter} & \textbf{Kerr-AdS} & \textbf{Kerr-AdS} \\ 
 & & (Hawking-Reall $\checkmark$) & (Hawking-Reall $\times$) \\
 \hline
 \textbf{Reflective} & time-periodic solutions, & inverse logarithmic decay & exponentially growing \\
 conditions & no decay & \cite{holzegel2013decay, holzegel2014quasimodes} & modes \cite{dold2017unstable} \\
 \hline
 \textbf{Dissipative} & $1/t^n$ decay for all $n$ & $1/t^n$ decay for all $n$ when $a=0$ & ?\\
 conditions  & \cite{holzegel2020asymptotic} & (this paper), unknown for $a\neq0$ & 
\end{tabular}}
\caption{Energy decay rate for solutions to the conformal wave equation.}
\end{table}
We finally comment on the relevance of this result for the black hole stability problem. The system of linearized Einstein equations is controlled by Teukolsky equation \cite{teukolsky1973perturbations}, which, when transformed into Regge-Wheeler equation \cite{regge1957stability}, bears remarkable similarities with the conformal wave equation \eqref{WE}. The natural next step is therefore to study boundedness and decay for solutions to the Regge-Wheeler equation subject to dissipative boundary conditions, for which the present work suggests inverse polynomial decay rates — potentially sufficient to establish nonlinear stability of Schwarzschild-AdS black holes under dissipative conditions.
\subsection{Comparison with \texorpdfstring{$\Lambda\geq 0$}{Lambda}}
In sharp contrast with $\Lambda<0$, the analogous questions of decay for scalar perturbations have been thoroughly explored when $\Lambda\geq0$. Since the literature is extensive, we do not give a comprehensive account but rather highlight the key difference in decay rates between $\Lambda=0$ and $\Lambda>0$. In Schwarzschild, boundedness and integrated decay on Schwarzschild were obtained in \cite{kay1987linear, dafermos2009red, blue2003}, and extended to the full subextremal Kerr family in \cite{dafermos2010decay, dafermos2016decay}. Pointwise decay estimates on Schwarzschild were established in \cite{donninger2012pointwise}, and on a general class of asymptotically flat spacetimes including subextremal Kerr in \cite{tataru2013local}; almost-sharp decay via physical-space methods was obtained in \cite{angelopoulos2018vector}. In fact, the late-time behavior of scalar perturbations is governed by Price's law \cite{price1972nonspherical}, which predicts a fixed inverse polynomial decay rate $\psi_\ell \sim t^{-3-2\ell}$, where $\ell$ denotes the spherical harmonic mode, along constant $r$ hypersurfaces; this was rigorously established in \cite{angelopoulos2018late}. For similar results in the full subextremal Kerr family, see \cite{hintz2022sharp, angelopoulos2023late}. In asymptotically de Sitter spacetimes the presence of a cosmological horizon produces an exponential redshift mechanism that drives solutions to exponential decay \cite{hintz2022sharp}.

\subsection{Outline of the paper}\label{subsec:outline}
In Section~\ref{sec:boundedness} we prove Theorem~\ref{thm:boundednessintro}. 
In Subsection~\ref{subsec:tenergy} we derive the energy conservation law for 
the degenerate energy $\mathbb{E}^T[\psi]$ and deduce energy boundedness 
(Proposition~\ref{prop:Tenergy}). In Subsection~\ref{subsec:redshift} we 
establish the redshift estimate (Proposition~\ref{prop:redshift}), which 
provides integrated decay near $\mathcal{H}^+$ conditional on integrated decay 
away from it. In Subsection~\ref{subsec:Nenergy} we combine the two to obtain 
non-degenerate energy boundedness for $\overline{\mathbb{E}}[\psi]$.

In Section~\ref{sec:decay} we prove Theorems~\ref{thm:integrateddecayintro} 
and~\ref{thm:deg-integrated} and Corollary~\ref{cor:decay}. In 
Subsection~\ref{subsec:degdecay} we apply the Morawetz multiplier 
\eqref{morawetzmultiplier} to obtain a degenerate integrated decay estimate 
for $T\psi$ (Theorem~\ref{thm:deg-integrated}), which in particular yields 
non-degenerate control of $T(r\psi)$. In Subsection~\ref{subsec:elliptic} 
we extend control to the spatial derivatives and to $r\psi$ itself via an 
elliptic estimate on the elliptic part of $\square_g$, completing the proof 
of Theorem~\ref{thm:integrateddecayintro}. In 
Subsection~\ref{subsec:polydecay} we promote the integrated decay estimate 
to arbitrarily fast inverse polynomial decay in $v$, proving 
Corollary~\ref{cor:decay}.

In the Appendix, we establish positivity of the zeroth order coefficient $c_0$ of the bulk term arising in the proof of Theorem~\ref{thm:deg-integrated} using a Hardy-type inequality.

\subsection*{Acknowledgements}
I thank my supervisor, Gustav Holzegel, for introducing me to this problem and guiding me through it with many helpful and interesting conversations. I acknowledge support by the Deutsche Forschungsgemeinschaft (DFG, German Research Foundation) under Germany's Excellence Strategy EXC 2044/2 -- 390685587, Mathematics Münster: Dynamics-Geometry-Structure, and from the Alexander von Humboldt Foundation in the framework of the Alexander von Humboldt Professorship endowed by the Federal Ministry of Education.

\section{Energy boundedness}\label{sec:boundedness}
In this section we prove Theorem \ref{thm:boundednessintro}. In Subsection \ref{subsec:tenergy} we prove the first statement, that is the energy conservation law and energy boundedness for the (degenerate at $\Hcal^+$) energy $\mathbb{E}^T[\psi]$. In Subsection \ref{subsec:redshift} we use the redshift effect to prove an accessory result, Proposition \ref{prop:redshift}, which we employ in Subsection \ref{subsec:Nenergy} to deduce energy boundedness for the non-degenerate energy $\overline{\mathbb{E}}$ and also invoke later, in the proof of Theorem \ref{thm:integrateddecayintro}.

\subsection{Energy boundedness for \texorpdfstring{$\mathbb{E}^T[\psi]$}{T energy} }\label{subsec:tenergy}
Statement 1 of Theorem \ref{thm:boundednessintro} is equivalent to the proposition below.
\begin{prop}\label{prop:Tenergy}
Let $\psi$ be any smooth solution to \eqref{WE} subject to dissipative boundary conditions \eqref{BC} arising from initial data as in Theorem \ref{thm:wp}. For any $v_0\leq v_1\leq v_2$, we have the energy conservation law:
\[ \mathbb{E}^T[\psi](v_2)+\int\limits_{\scri^{+}(v_1,v_2)}(T(r\psi))^{2}dvd\omega+\int\limits_{\mathcal{H}^{+}(v_1,v_2)}(T(r\psi))^{2}dvd\omega=\mathbb{E}^T[\psi](v_1).\]
Therefore, the energy $\mathbb{E}^T[\psi]$ is non-increasing in $v$ because $\mathbb{E}^T[\psi](v_2)\leq\mathbb{E}^T[\psi](v_1)$. 
\end{prop}
\proof
Multiply \eqref{WErpsi} by $T(r\psi)$ and integrate by parts, to find:
\begin{multline*}
T\left[-\dfrac{1}{2}(T(r\psi))^{2}-\dfrac{1}{2}(R^{*}(r\psi))^{2}-\dfrac{1}{2}\dfrac{1-\frac{2M}{r}+\frac{r^2}{l^2}}{r^{2}}\abs{r\nablas(r\psi)}^{2}-\dfrac{1}{2}V(r)(r\psi)^{2}\right] + \\
R^{*}\left[T(r\psi)R^{*}(r\psi)\right]=0.
\end{multline*}
Using $T=\de_{v}$ and $R^{*}=\de_{v}+\Big(1-\frac{2M}{r}+\frac{r^2}{l^2}\Big)\de_{\tilde{r}}$, we find
\begin{multline*}
\de_{v}\left[-\dfrac{1}{2}((T-R^{*})(r\psi))^{2}-\dfrac{1}{2}\dfrac{1-\frac{2M}{r}+\frac{r^2}{l^2}}{r^{2}}\abs{r\nablas(r\psi)}^{2}-\dfrac{1}{2}V(r)(r\psi)^{2}\right]+ \\ 
+\Big(1-\frac{2M}{r}+\frac{r^2}{l^2}\Big)\de_{\tilde{r}}\left[T(r\psi)R^{*}(r\psi)\right]=0.
\end{multline*}
Integrating the equation in $\Big(1-\frac{2M}{r}+\frac{r^2}{l^2}\Big)^{-1}dvdrd\omega$ over the region $[v_1,v_2]\times[r_+,\infty)\times\sfera^2$ and using dissipative boundary conditions, we find the desired conservation law.
\endproof
We may then deduce a pointwise boundedness statement by use of the standard Sobolev embedding for the $2$-sphere, namely $H^2(\sfera^2)\hookrightarrow L^\infty(\sfera^2)$. 
\begin{cor}
    There exists a constant $C>0$ only depending on the black hole parameters such that the pointwise bound
\[ |\psi| \leq C\cdot  \dfrac{\sqrt{\sum\limits_{k=0}^{2}\sum\limits_{|\alpha|=k}\mathbb{E}^T[\Gamma^{\alpha}\psi](v_0)}}{r} \]
 holds to the future of and including $\Sigma_{v_0}$, in the exterior of the black hole.
\end{cor}
\begin{remark}\label{rem:JTcurrent}
In principle, the energy conservation law for $\mathbb{E}^T[\psi]$ can be obtained by applying Stokes' Theorem to the current $\mathbb{J}^T[\psi]$\footnote{See Subsubsection \ref{subsubsec:vf-notation} for the notation used in this remark.}. The resulting energy along null hypersurfaces $\Sigma_v$ would look like:
\begin{multline}
\int\limits_{\Sigma_v}\mathbb{J}^T[\psi](-n_{\Sigma_v})dvol_{\Sigma_v}= \\ =\int\limits_{\Sigma_v}\mathbb{T}(T,-\de_{\tilde{r}})r^2drd\omega=\frac12\int\limits_{\Sigma_v}\left[\Big(1-\frac{2M}{r}+\frac{r^2}{l^2}\Big)(\de_{\tilde{r}}(r\psi))^2+\abs{\nablas\psi}^2-\frac{2}{l^2}\psi^2\right]r^2drd\omega
\end{multline}
where we used the normalized couple
\[  n_{\Sigma_{v}} =(dv)^{\#}=\de_{\tilde{r}} \quad \text{and} \quad
dvol_{\Sigma_{v}} =r^{2}drd\omega \quad \text{with} \quad  d\omega=\sin\theta d\theta d\varphi,\]
Notice how the energy, in addition to not being coercive because of the zeroth order term, is not manifestly finite because dissipative boundary conditions imply a decay rate of $r^{-1}$ at $\scri^+$. These issues are resolved in two steps:
\begin{enumerate}
\item firstly, one rewrites the integrand, using the $r^2$ weight coming from the volume form, as an energy for $r\psi$ instead of $\psi$, and thus resolves the sign issue;
\item secondly, during the application of Stokes' Theorem, cancellations between the fluxes along $\Sigma_v$ hypersurfaces and the flux along $\scri^+$ resolve the integrability issue\footnote{This cancellation only happens in the conformal case. If we consider the massive wave equation $\square_g\psi+\frac{\alpha}{l^2}\psi=0$ for other values of the mass parameter $\alpha$, cancellation does not occur.} and thus reveal which finite energy to use for the conservation law.
\end{enumerate}
Were we to consider Dirichlet boundary conditions, we would have neither of the issues above because the decay rate at infinity would be $r^{-2}$, which guarantees integrability and also allows one to use a Hardy-type inequality to absorb the negative zeroth order term using the $(\de_{\tilde{r}}\psi)^2$ term -- see Lemma 4.1 in \cite{holzegel2010massive}.
\end{remark}
We shift our focus towards the zeroth order term appearing in $\mathbb{E}^T[\psi](v)$, whose weight of $r^{-3}$ at infinity we may improve to $r^{-2}$ by borrowing from $(\de_{\tilde{r}}(r\psi))^{2}$.
\begin{lem}[Weight improvement]\label{lem:weightimprovement}
Let $\chi(r)$ be a smooth cut-off function monotonically increasing from $0$ to $1$ over a bounded region $[r_{0},r_{1}]$, with $r_{+}<r_{0}<r_{1}<\infty$. Then
\[
    \int\limits_{\Sigma_{v}}\chi(r)\dfrac{(r\psi)^{2}}{r^{2}}\,dr\,d\omega \lesssim \mathbb{E}^T[\psi](v).
\]
\end{lem}

\begin{proof}
Integration by parts is sufficient to conclude the proof:
\begin{align*}
    \int\limits_{\Sigma_{v}}\chi(r)\dfrac{1}{r^{2}}(r\psi)^{2}\,dr\,d\omega
    = \int\limits_{\Sigma_{v}}\chi(r)\de_{\tilde{r}}\!\left(-\dfrac{1}{r}\right)(r\psi)^{2}\,dr\,d\omega \\
    = \left.-\dfrac{\chi(r)}{r}(r\psi)^{2}\right\rvert_{\mathcal{H}^{+}}^{\scri^{+}}
      + \int\limits_{\Sigma_{v}}\dfrac{1}{r}\left(2\chi(r)(r\psi)(\de_{\tilde{r}}(r\psi))
      + \de_{\tilde{r}}\chi(r)(r\psi)^{2}\right)dr\,d\omega \\
    \leq \int\limits_{\Sigma_v}\left[
        \left(\dfrac{\de_{\tilde{r}}\chi(r)}{r}
        + \dfrac{\chi(r)}{r\!\left(1-\frac{2M}{r}+\frac{r^2}{l^2}\right)}\right)(r\psi)^2
        + \dfrac{\chi(r)}{r}\!\left(1-\frac{2M}{r}+\frac{r^2}{l^2}\right)(\de_{\tilde{r}}(r\psi))^2
      \right]dr\,d\omega \\
    \lesssim \mathbb{E}^T[\psi](v).
\end{align*}
\end{proof}

\noindent Consider the improved energy
\begin{equation}
    \tilde{\mathbb{E}}^T[\psi](v) := \dfrac{1}{2}\int\limits_{\Sigma_{v}}\left[
        \left(1-\frac{2M}{r}+\frac{r^2}{l^2}\right)(\de_{\tilde{r}}(r\psi))^{2}
        + \abs{\nablas(r\psi)}^{2}
        + \dfrac{(r\psi)^{2}}{r^{2}}
    \right]dr\,d\omega.
\end{equation}
Using the lemma above, we obtain the following.

\begin{lem}
The energies $\mathbb{E}^T[\psi]$ and $\tilde{\mathbb{E}}^T[\psi]$ are pointwise equivalent in $v$, in the sense that, for any $v$, we have
\begin{equation}
    \tilde{\mathbb{E}}^T[\psi](v) \lesssim \mathbb{E}^T[\psi](v) \lesssim \tilde{\mathbb{E}}^T[\psi](v).
\end{equation}
Therefore, for any $v_0\leq v_1 \leq v_2$, the energy $\tilde{\mathbb{E}}^T[\psi](v)$ satisfies
\begin{equation}
    \tilde{\mathbb{E}}^T[\psi](v_2) \lesssim \tilde{\mathbb{E}}^T[\psi](v_1).
\end{equation}
\end{lem}

\begin{proof}
The proof of the first statement follows directly from Lemma~\ref{lem:weightimprovement}. Then, for any $v_0\leq v_1 \leq v_2$, using that $\mathbb{E}^T[\psi]$ is non-increasing in $v$, as proven in Proposition~\ref{prop:Tenergy} above, we have
\[
    \tilde{\mathbb{E}}^T[\psi](v_2)
    \lesssim \mathbb{E}^T[\psi](v_2)
    \leq \mathbb{E}^T[\psi](v_1)
    \lesssim \tilde{\mathbb{E}}^T[\psi](v_1).
\]
\end{proof}
\noindent In light of the above, we may use the energy $\mathbb{E}^T[\psi](v)$ to control $(r\psi)^2/r^2$ along $\Sigma_v$. We avoid mentioning $\tilde{\mathbb{E}}^T[\psi]$ to ease the notation. 
\begin{remark}\label{rem:controllingTrpsi}
Since we integrate along null hypersurfaces $\Sigma_v$, the energy $\mathbb{E}^T[\psi]$ does not feature the transversal $T(r\psi)$ derivative. Were we to integrate along a spacelike hypersurface $\tilde{\Sigma}_{\tau}$, this would no longer be the case. The goal of this remark is to compare the energy on the null slice $\Sigma_{v_0}$ to the corresponding energy flux on a spacetime slice $\tilde{\Sigma}_\tau$ to its future, to establish that
\[ \int\limits_{\Sigma_\tau}\frac{(T(r\psi))^2}{r^2}drd\omega \leq \mathbb{E}^T[\psi](v_0).\] 
Fix $v_0\leq v_1$ and foliate the region between $\Sigma_{v_0}$ and $\Sigma_{v_1}$ with spacelike hypersurfaces $\tilde{\Sigma}_\tau$ such that $\tau$ is like $t^*=t+r^*-l\arctan(r/l)$ near $\Hcal^+$ and like $t$ near $\scri^+$. Then, let $\tau_0$ be such that $\Sigma_{v_0}\cap \scri^+$ is the sphere at infinity where $\tau=\tau_0$, and $\tau_1$ be such that $\Sigma_{v_1}\cap\Hcal^+$ is the sphere at the horizon where $\tau=\tau_1$. For any $\tau\in[\tau_0,\tau_1]$, consider the region bounded by $\Sigma_{v_0}$ at the bottom and $\tilde{\Sigma}_\tau$ at the top -- see Figure \ref{fig:spacelikevsnull} below. Applying the same integration by parts used in the proof of Proposition \ref{prop:Tenergy} over this region, one obtains the energy conservation law
\[ \mathcal{E}^T[\psi](\tau)+\int\limits_{\Hcal^+(v_0,v_{\Hcal^+}(\tau))}(T(r\psi))^2dvd\omega+\int\limits_{\scri^+(v_0,v_{\scri^+}(\tau) )}(T(r\psi))^2dvd\omega=\mathbb{E}^T[\psi](v_0),\]
where $\mathcal{E}^T[\psi](\tau)$ is obtained from the energy flux along $\tilde{\Sigma}_\tau$ and therefore features $(T(r\psi))^2/r^2$. We use notation $v_{\Hcal^+}(\tau)$ to denote the $v$ coordinate on the sphere where $\tilde{\Sigma}_{\tau}$ intersects $\Hcal^+$ and $v_{\scri^+}(\tau)$ on the sphere where it intersects $\scri^+$. One immediately obtains that
 \[\mathcal{E}^T[\psi](\tau)\leq \mathbb{E}^T[\psi](v_0),\]
which confirms that $\mathbb{E}^T[\psi]$ controls the transversal derivative as well. Furthermore, for any $\tau_0\leq \tau_1$ such that $\tilde{\Sigma}_{\tau_0}$ is to the future of $\Sigma_{v_0}$, we observe that
\[ \int_{\tau_0}^{\tau_1}\mathcal{E}^T[\psi](\tau)d\tau \leq (\tau_1-\tau_0)\mathbb{E}^T[\psi](v_0) \leq (v_1-v_0+\frac{\pi}{2})\mathbb{E}^T[\psi](v_0),\]
where the constant $\frac{\pi}{2}$ appears because of how we normalized the tortoise coordinate at $\scri^+$. 
\begin{figure}[H]
\centering
\scalebox{1.2}{
\begin{tikzpicture}
    \node (I) at (0,0) {};
    \node (II) at (2,-2) {\raisebox{-6mm}{$i^-$}};
    \node (III) at (2,2) {\raisebox{+6mm}{$i^+$}};

    \fill (I) circle (0.5pt);
    \fill (II) circle (0.5pt);
    \fill (III) circle (0.5pt);
    
    \draw[dashed] (I.center)--(II.center) node[midway, below, xshift=-2mm] {$\Hcal^-$};
    \draw[dashed] (II.center)--(III.center) node[midway, xshift=4mm] {$\scri^+$}; 
    \draw (III.center)--(I.center) node[midway, above, xshift=-2mm] {$\Hcal^+$};

    \draw[thick] (0.2,0.2)--(2,-1.6) node[pos=0, above, xshift=-2mm] {{\footnotesize $\Sigma_{v_0}$}};
    \draw[thick] (1.5,1.5)--(2,1) node[pos=1, above, xshift=4mm] {{\footnotesize $\Sigma_{v_1}$}};

    \draw[gray!60] (1,1)--(2,1);
    \draw[gray!60] (1.2,1.2)--(1.8,1.2);
    
    \draw[gray!60] (0.2,0.2)--(2,0.2);
    \draw[gray!60] (0.4,0.4)--(2,0.4);
    \draw[gray!60] (0.6,0.6)--(2,0.6);
    \draw[gray!60] (0.8,0.8)--(2,0.8);
    
    \draw[gray!60] (0.4,0)--(2,0);
    \draw[gray!60] (0.6,-0.2)--(2,-0.2);
    \draw[gray!60] (0.8,-0.4)--(2,-0.4);
    \draw[gray!60] (1,-0.6)--(2,-0.6);
    \draw[gray!60] (1.2,-0.8)--(2,-0.8);
    \draw[gray!60] (1.4,-1)--(2,-1);
    \draw[gray!60] (1.6,-1.2)--(2,-1.2);
    
    \filldraw[black] (1.5, 1.5) circle (1.5pt);
\node at (1.5, 1.5) [above left] {\footnotesize $\tau_1$};
	\filldraw[black] (2, -1.6) circle (1.5pt);
\node at (2, -1.6) [above right] {\footnotesize $\tau_0$};
\end{tikzpicture}
\hspace{2cm}
\begin{tikzpicture}
    \node (I) at (0,0) {};
    \node (II) at (2,-2) {\raisebox{-6mm}{$i^-$}};
    \node (III) at (2,2) {\raisebox{+6mm}{$i^+$}};

    \fill (I) circle (0.5pt);
    \fill (II) circle (0.5pt);
    \fill (III) circle (0.5pt);
    
    \draw[dashed] (I.center)--(II.center) node[midway, below, xshift=-2mm] {$\Hcal^-$};
    \draw[dashed] (II.center)--(III.center) node[midway, xshift=4mm] {$\scri^+$}; 
    \draw (III.center)--(I.center) node[midway, above, xshift=-2mm] {$\Hcal^+$};

    \draw[blue, thick] (0.2,0.2)--(2,-1.6) node[pos=0, above, xshift=-2mm] {{\footnotesize $\Sigma_{v_0}$}};
    \draw[thick] (1.5,1.5)--(2,1) node[pos=1, above, xshift=4mm] {{\footnotesize $\Sigma_{v_1}$}};

    \draw[thick, blue] (0.6,0.6)--(2,0.6) node[pos=0.6, above] {\footnotesize $\tilde{\Sigma}_\tau$};

   \fill[cyan!50, opacity=0.3] (0.2,0.2)--(0.6,0.6)--(2,0.6)--(2,-1.6)--cycle;

\end{tikzpicture}}

\caption{Left: schematic representation of the spacelike foliation. Right: integration region.}
\label{fig:spacelikevsnull}
\end{figure}
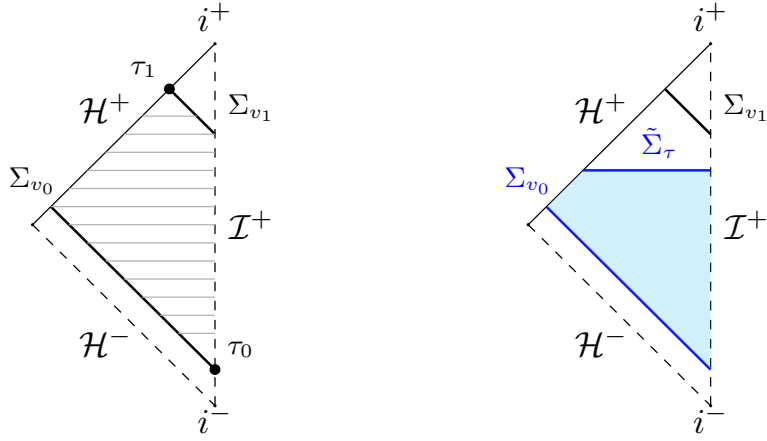
\end{remark}

\subsection{Redshift Estimate}\label{subsec:redshift}
Ever since \cite{dafermos2008lectures, dafermos2009red}, it is well-known that the loss of control of $(\de_{\tilde{r}}(r\psi))^2$ near $\Hcal^+$ in $\mathbb{E}^T[\psi]$ can be avoided by use of a redshift vector field $N$, which captures the homonymous phenomenon. In this subsection, we use a redshift vector field to establish the following result, which we employ in Subsection \ref{subsec:Nenergy} to eliminate the degeneracy in $\mathbb{E}^T[\psi]$ and also invoke later, in the proof of integrated decay, to obtain control near $\Hcal^+$.
\begin{prop}[Redshift Estimate]\label{prop:redshift}
Let $\psi$ be any smooth solution to \eqref{WE} subject to dissipative boundary conditions \eqref{BC} arising from initial data as in Theorem \ref{thm:wp}. There exist numbers $r_{\text{RED}}$ and $r_{\text{CUT}}$, depending only on $M$ and $l$ and with $r_+<r_{\text{RED}}<r_{\text{CUT}}$, such that, for any $v_0\leq v_1\leq v_2$, we have:
    \begin{align}\label{redshift}
    \int\limits_{v_1}^{v_2}\int\limits_{r_+}^{r_{\text{RED}}}\int\limits_{\sfera^2}\Big[(T(r\psi))^2+(\de_{\tilde{r}}(r\psi))^2+\abs{\nablas(r\psi)}^2+(r\psi)^2 \Big]dvdrd\omega +\overline{\mathbb{E}}[\psi](v_2) \lesssim \nonumber \\  \int\limits_{v_1}^{v_2}\int\limits_{{r_{\text{RED}}}}^{{r_{\text{CUT}}}}\int\limits_{\sfera^2}\Big[(T(r\psi))^2+(\de_{\tilde{r}}(r\psi))^2+\abs{\nablas(r\psi)}^2+(r\psi)^2\Big]dvdrd\omega + \overline{\mathbb{E}} [\psi](v_1).
    \end{align}
\end{prop}
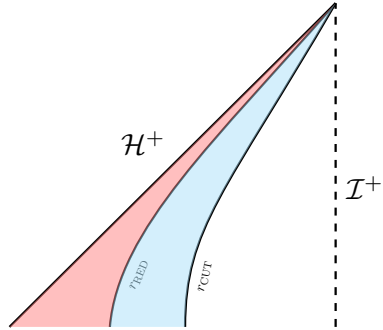
\begin{figure}[H]
\centering
\scalebox{1}{
\begin{tikzpicture} 
\tikzset{x=0.7ex,y=0.7ex} 

    \clip (0,0) rectangle (85.50, 49.01); 

    \draw[line width=0.85pt] (14.23, 6.32) -- (51.23, 43.32) node[midway, above, xshift=-1em] {$\Hcal^+$};
    \draw[dashed, line width=0.85pt] (51.23, 6.32) -- (51.23, 43.32) node[midway, below, xshift=1em] {$\scri^+$} ;
    \draw[line width=0.85pt] (25.62, 6.32) .. controls (26.60, 17.57) and (37.64, 27.64) .. (51.23, 43.32)  node[pos=0.2, below, sloped] {\scalebox{0.5}{$r_{\text{RED}}$}};
    \draw[line width=0.85pt] (34.15, 6.32) .. controls (34.15, 17.70) and (38.73, 21.69) .. (51.23, 43.32) node[pos=0.2, below, sloped] {\scalebox{0.5}{$r_{\text{CUT}}$}};

    \fill[red!50, opacity=0.5] (14.23,6.32) -- (25.62,6.32) .. controls (26.60,17.57) and (37.64,27.64) .. (51.23,43.32) -- (51.23,43.32) -- cycle;
    \fill[cyan!30, opacity=0.5]  (25.62,6.32)  .. controls (26.60,17.57) and (37.64,27.64) .. (51.23,43.32) .. controls (38.73,21.69) and (34.15,17.70) .. (34.15,6.32) -- cycle;
    
\end{tikzpicture}}
\caption{The redshift estimate guarantees integrated decay in the red region (near $\Hcal^+$) provided that we have integrated decay in the blue region (away from $\Hcal^+$).}
\end{figure}
\subsubsection{Accessory notation \& lemmata}\label{subsubsec:vf-notation}
Recall that, if $\psi$ solves \eqref{WE}, the associated stress-energy tensor
\[ \mathbb{T}_{\mu\nu}[\psi]:=\de_{\mu}\psi\de_{\nu}\psi-\dfrac{1}{2}g_{\mu\nu}(g^{\alpha\beta}\de_{\alpha}\psi\de_{\beta}\psi-\frac{2}{l^{2}}\psi^{2})\]
is divergence free: $\nabla^{\mu}\mathbb{T}_{\mu\nu}=0$. We use $\mathbb{T}_{\mu\nu}[\psi]$ to define currents such as $\mathbb{J}^{X}_{\mu}[\psi]:=\mathbb{T}_{\mu\nu}[\psi]X^{\nu}$, where $X$ is some vector field. Also recall that $\nabla^{\mu}\mathbb{J}^{X}_{\mu}[\psi]=\mathbb{T}_{\mu\nu}{^{(X)}}\pi^{\mu\nu}$, where ${^{(X)}}\pi$ denotes the deformation tensor associated to $X$, which is defined using the Lie derivative as ${^{(X)}}\pi=\mathcal{L}_{X}g$ and thus vanishes when $X$ is a Killing vector.
We use the notation
\[ \mathbf{K}_0^X[\psi]=\mathbb{T}_{\mu\nu}{^{(X)}}\pi^{\mu\nu}, \quad 
    \mathbf{K}_{\star}^X[\psi]= \Big(\de_{\mu}\psi\de_{\nu}\psi-\dfrac{1}{2}g_{\mu\nu}\de^{\alpha}\psi\de_{\alpha}\psi\Big){^{(X)}}\pi^{\mu\nu}.\]
Note that $\mathbf{K}_{\star}^X[\psi]$ is obtained from the divergence $\mathbf{K}_0^X[\psi]$ by removing the zeroth order term. In particular, it is a symmetric bilinear form with respect to the first order derivatives $\de\psi$ of $\psi$.

We now state two lemmata that shall be useful later in the proof of Proposition \ref{prop:redshift}.
\begin{lem}\label{lem:LemmaH}
Let $\psi:\R\times[r_+,\infty)\times\sfera^2\longrightarrow\R$ be a smooth function. Let $v_0\leq v_1\leq v_2$. There exists a constant $C=C(M,L)$ that only depends on black hole parameters such that, for any $\epsilon>0$, we have that
    \[ \int\limits_{\Hcal^+(v_1,v_2)}(r\psi)^2dvd\omega\leq \epsilon\int\limits_{v_1}^{v_2}\int\limits_{r_+}^{r_{\text{CUT}}}\int\limits_{\sfera^2}(\de_{\tilde{r}}(r\psi))^2dvdrd\omega+\Big(\frac{1}{C}+\frac{1}{\epsilon}\Big)\int\limits_{v_1}^{v_2}\int\limits_{r_+}^{r_{\text{CUT}}}\int\limits_{\sfera^2}(r\psi)^2dvdrd\omega. \]
\end{lem}
\begin{proof}
    Let $\rho:[r_+,\infty)\longrightarrow\R$ be the smoothening of a piecewise $C^1$ cut-off function decreasing monotonically from $1$ to $0$ in $[r_+,r_{\text{CUT}}]$. Then, by the Fundamental Theorem of Calculus and Cauchy-Schwarz inequality, we have that
\begin{align}
\int\limits_{\Hcal^+(v_1,v_2)}(r\psi)^2dvd\omega=\int\limits_{v_1}^{v_2}\int\limits_{r_+}^{\infty}\int\limits_{\sfera^2}-\de_{\tilde{r}}(\rho(r\psi)^2)dvdrd\omega\leq \int\limits_{v_1}^{v_2}\int\limits_{r_+}^{\infty}\int\limits_{\sfera^2}\Big|\de_{\tilde{r}}(\rho(r\psi)^2)\Big|dvdrd\omega \leq  \nonumber \\ 
\leq \int\limits_{v_1}^{v_2}\int\limits_{r_+}^{\infty}\int\limits_{\sfera^2}\epsilon\rho(\de_{\tilde{r}}(r\psi))^2)dvdrd\omega+\int\limits_{v_1}^{v_2}\int\limits_{r_+}^{\infty}\int\limits_{\sfera^2}\Big(\abs{\de_{\tilde{r}}\rho}+\frac{\rho}{\epsilon}\Big)(r\psi)^2dvdrd\omega
\end{align}
Since $\max\rho=1$ and $\max\abs{\de_{\tilde{r}}\rho}=\frac{1}{r_{\text{CUT}}-r_+}$, we obtain:
\[ \int\limits_{\Hcal^+(v_1,v_2)}(r\psi)^2dvd\omega\leq \epsilon\int\limits_{v_1}^{v_2}\int\limits_{r_+}^{r_{\text{CUT}}}\int\limits_{\sfera^2}(\de_{\tilde{r}}(r\psi))^2dvdrd\omega+\Big(\frac{1}{r_{\text{CUT}}-r_+}+\frac{1}{\epsilon}\Big)\int\limits_{v_1}^{v_2}\int\limits_{r_+}^{r_{\text{CUT}}}\int\limits_{\sfera^2}(r\psi)^2dvdrd\omega \]
Choosing $C(M,L)=1/(r_{\text{CUT}}-r_+)$, we have the desired estimate.
\end{proof}
\begin{lem} \label{lem:0-near-away}
Let $\psi:\R\times[r_+,\infty)\times\sfera^2\longrightarrow\R$ be a smooth function. Then, for any $\eta>0$, we have
\[ \int\limits_{r_+}^{r_{\text{RED}}}\int\limits_{\sfera^2}(r\psi)^2dvdrd\omega\leq \eta \int\limits_{r_+}^{r_{\text{RED}}}\int\limits_{\sfera^2}(\de_{\tilde{r}}(r\psi))^2dvdrd\omega+\frac{1}{\eta}\int\limits_{r_{\text{RED}}}^{r_{\text{CUT}}}\int\limits_{\sfera^2}\Big((\de_{\tilde{r}}(r\psi))^2+(r\psi)^2)\Big)dvdrd\omega \] 
\end{lem}
\begin{proof}
Let $u\in C^\infty([r_+,\infty))$ and $r\in[r_+,r_{\text{RED}}]$. Using the Fundamental Theorem of Calculus, triangular inequality and Cauchy-Schwarz inequality we have
\begin{equation}
    |u(r)|^2\leq C_1|u(r_{\text{RED}})|^2+\dfrac{(r_{\text{RED}}-r_+)}{C_1}\int\limits_{r_+}^{r_{\text{RED}}}|u'(t)|^2dr,
\end{equation}
for any positive constant $C_1>0$. Integrate over $[r_+,r_{\text{RED}}]$ to find
\begin{equation}\label{eq:lemma0-1}
    \int\limits_{r_+}^{r_{\text{RED}}}|u(r)|^2dr\leq C_1(r_{\text{RED}}-r_+)|u(r_{\text{RED}})|^2+\dfrac{(r_{\text{RED}}-r_+)^2}{C_1}\int\limits_{r_+}^{r_{\text{RED}}}|u'(r)|^2dr.
\end{equation}
Now we estimate $u(r_{\text{RED}})$ in an entirely analogous manner. Let $s\in[r_{\text{RED}},r_{\text{CUT}}]$. Then
\begin{equation}
    |u(r_{\text{RED}})|^2\leq C_2|u(s)|^2+\dfrac{s-r_{\text{RED}}}{C_2}\int\limits_{r_{\text{RED}}}^s|u'(t)|^2dt.
\end{equation}
Integrate over $[r_{\text{RED}},r_{\text{CUT}}]$ and obtain
\begin{align}
    (r_{\text{CUT}}-r_{\text{RED}})|u(r_{\text{RED}})|^2& \leq  C_2\int\limits_{r_{\text{RED}}}^{r_{\text{CUT}}}|u(s)|^2ds+\int\limits_{r_{\text{RED}}}^{r_{\text{CUT}}}\dfrac{(s-r_{\text{RED}})}{C_2}\int\limits_{r_{\text{RED}}}^s|u'(t)|^2dt ds\nonumber \\
    & \leq C_2\int\limits_{r_{\text{RED}}}^{r_{\text{CUT}}}|u(s)|^2ds+\dfrac{(r_{\text{CUT}}-r_{\text{RED}})^2}{C_2}\int\limits_{r_{\text{RED}}}^{r_{\text{CUT}}}|u'(t)|^2dt; \nonumber \\
    |u(r_{\text{RED}})|^2 &\leq \dfrac{C_2}{r_{\text{CUT}}-r_{\text{RED}}}\int\limits_{r_{\text{RED}}}^{r_{\text{CUT}}}|u(s)|^2ds+\dfrac{r_{\text{CUT}}-r_{\text{RED}}}{C_2}\int\limits_{r_{\text{RED}}}^{r_{\text{CUT}}}|u'(t)|^2dt. \label{eq:lemma0-2}
\end{align}
We choose $C_2=r_{\text{CUT}}-r_{\text{RED}}$ and eliminate dependency on $r_{\text{CUT}}-r_{\text{RED}}$. Then, we plug equation \eqref{eq:lemma0-2} into \eqref{eq:lemma0-1} and obtain
\begin{equation}
    \int\limits_{r_+}^{r_{\text{RED}}}|u|^2dr\leq \dfrac{(r_{\text{RED}}-r_+)^2}{C_1}\int\limits_{r_+}^{r_{\text{RED}}}|u'|^2dr+C_1(r_{\text{RED}}-r_+)\int\limits_{r_{\text{RED}}}^{R_2}(|u'|^2+|u|^2)dr.
\end{equation}
Now let $u=r\psi$ and observe that, for any $\eta>0$, we can choose $C_1$ so that the desired inequality holds true.
\end{proof}
\subsubsection{Proof of Proposition \ref{prop:redshift}}

\begin{proof}[Proof of Proposition \ref{prop:redshift}]
By \cite{dafermos2008lectures}, there exist numbers $r_{\text{RED}}$ and $r_{\text{CUT}}$ that only depend on the black hole parameters $M$ and $l$, with $r_+<r_{\text{RED}}<r_{\text{CUT}}$, and a vector field $N$\footnote{See also \cite{holzegel2010massive} for an explicit construction of $N$ in Schwarzschild-AdS.} with the following properties:
\begin{enumerate}
\item $N$ is future-directed timelike and invariant along the integral curves of $T$;
\item for all $r\geq r_{\text{CUT}}$, we have $N\equiv 0$;
\item the bilinear form $\mathbf{K}^N_*[\psi]$ is such that
\begin{align}
    \mathbf{K}_{\star}^N[\psi]r^2\gtrsim & \Big[(T(r\psi))^2+(\de_{\tilde{r}}(r\psi))^2+\abs{\nablas(r\psi)}^2\Big] & & \text{when} \ r_+\leq r\leq r_{\text{RED}}, \label{KnearH} \\
    -\mathbf{K}_{\star}^N[\psi]r^2\lesssim & \Big[(T(r\psi))^2+(\de_{\tilde{r}}(r\psi))^2+\abs{\nablas(r\psi)}^2\Big] & & \text{when} \ r_{\text{RED}}\leq r\leq r_{\text{CUT}}. \label{KawayH}
\end{align}
\end{enumerate}
Apply Stokes' Theorem over the region $[v_1,v_2]\times[r_+,\infty)\times\sfera^2$ to the current $\mathbb{J}^N_\mu[\psi]$ and obtain:
\begin{align} \label{eq:stokes}
\int\limits_{v_1}^{v_2}\int\limits_{r_+}^{r_{\text{RED}}}\int\limits_{\sfera^2}\mathbf{K}_0^N[\psi]r^2dvdrd\omega +\int\limits_{\Hcal^+(v_1,v_2)}\mathbb{J}^N[\psi]n_{\Hcal^+}dvol_{\Hcal^+}+\int\limits_{\Sigma_{v_2}}\mathbb{J}^N[\psi](-n_{\Sigma_v})dvol_{\Sigma_v} = \nonumber \\
= \int\limits_{v_1}^{v_2}\int\limits_{r_{\text{RED}}}^{r_{\text{CUT}}}\int\limits_{\sfera^2}-\mathbf{K}_0^N[\psi]r^2dvdrd\omega+\int\limits_{\Sigma_{v_1}}\mathbb{J}^N[\psi](-n_{\Sigma_v})dvol_{\Sigma_v}
\end{align}
We refer to the spacetime integrals as bulk terms and to the integrals along the hypersurfaces $\Sigma_v, \Hcal^+(v_1,v_2),\scri^+(v_1,v_2)$ as boundary terms.
\vspace{3mm}

\textbf{Analysis of the boundary terms along} $\mathbf{\Sigma_v}$. By construction, the integrand $\mathbb{J}^N[\psi]n_{\Sigma_v}  $ is coercive if we exclude its zeroth order terms. However, since we control the zeroth order term with $\mathbb{E}^T[\psi](v)$, we can add as much as needed to obtain positivity. More precisely, there exists a constant $C>0$ that only depends on black hole parameters such that
\begin{equation}
\int\limits_{\Sigma_{v_2}}\mathbb{J}^N[\psi](-n_{\Sigma_v})dvol_{\Sigma_v}+C\cdot \mathbb{E}^T[\psi](v_2)\geq \overline{\mathbb{E}}[\psi](v_2).
\end{equation}
Since we also have that
\begin{equation}
\int\limits_{\Sigma_{v_1}}\mathbb{J}^N[\psi](-n_{\Sigma_v})dv + C\cdot \mathbb{E}^T[\psi](v_2) \leq \int\limits_{\Sigma_{v_1}}\mathbb{J}^N[\psi](-n_{\Sigma_v})dv + C\cdot \mathbb{E}^T[\psi](v_1) \lesssim \overline{\mathbb{E}}[\psi](v_1),
\end{equation}
we can add $C\cdot\mathbb{E}^T[\psi](v_1)$ to the left and right of the equal sign and arrive at
\begin{align}\label{eq:pre-redshift}
\int\limits_{v_1}^{v_2}\int\limits_{r_+}^{r_{\text{RED}}}\int\limits_{\sfera^2}\mathbf{K}_0^N[\psi]r^2dvdrd\omega +\int\limits_{\Hcal^+(v_1,v_2)}\mathbb{J}^N[\psi]\cdot n_{\Hcal^+}dvol_{\Hcal^+}+\overline{\mathbb{E}}[\psi](v_2) \lesssim \nonumber \\
\int\limits_{v_1}^{v_2}\int\limits_{r_{\text{RED}}}^{r_{\text{CUT}}}\int\limits_{\sfera^2}-\mathbf{K}_0^N[\psi]r^2dvdrd\omega +\overline{\mathbb{E}}[\psi](v_1).
\end{align}

\textbf{Analysis of the boundary term along} $\mathbf{\Hcal^+}$. By construction, the only negative contribution to the flux is given by its zeroth order term. We use Lemma \ref{lem:LemmaH} to control it with
\begin{equation}
\int\limits_{\Hcal^+(v_1,v_2)}(r\psi)^2dvd\omega\leq \epsilon\int\limits_{v_1}^{v_2}\int\limits_{r_+}^{r_{\text{CUT}}}\int\limits_{\sfera^2}(\de_{\tilde{r}}(r\psi))^2dvdrd\omega+\Big(\frac{1}{C_0}+\frac{1}{\epsilon}\Big)\int\limits_{v_1}^{v_2}\int\limits_{r_+}^{r_{\text{CUT}}}\int\limits_{\sfera^2}(r\psi)^2dvdrd\omega,
\end{equation}
where the constant $C_0$ only depends on black hole parameters and $\epsilon$ can be any positive number -- it will be chosen later.
\vspace{3mm}

\textbf{Analysis of the bulk terms.} After the last inequality, we arrive at 
\begin{align}
\int\limits_{v_1}^{v_2}\int\limits_{r_+}^{r_{\text{RED}}}\int\limits_{\sfera^2}\mathbf{K}_0^N[\psi]r^2dvdrd\omega+\overline{\mathbb{E}}[\psi](v_1) \lesssim \int\limits_{v_1}^{v_2}\int\limits_{r_{\text{RED}}}^{r_{\text{CUT}}}\int\limits_{\sfera^2}-\mathbf{K}_0^N[\psi]r^2dvdrd\omega +\overline{\mathbb{E}}[\psi](v_1)+ \nonumber \\
+ \epsilon\int\limits_{v_1}^{v_2}\int\limits_{r_+}^{r_{\text{CUT}}}\int\limits_{\sfera^2}(\de_{\tilde{r}}(r\psi))^2dvdrd\omega+\Big(\frac{1}{C_0}+\frac{1}{\epsilon}\Big)\int\limits_{v_1}^{v_2}\int\limits_{r_+}^{r_{\text{CUT}}}\int\limits_{\sfera^2}(r\psi)^2dvdrd\omega
\end{align}
The spacetime integral in the region $[r_+,r_{\text{RED}}]$ is positive if we exclude its zeroth order term -- recall property \eqref{KnearH} of $\mathbf{K}^N_*[\psi]$. Therefore, there exists a constant $C_1>0$, only depending on black hole parameters, such that
\begin{multline}
\int\limits_{v_1}^{v_2}\int\limits_{r_+}^{r_{\text{RED}}}\int\limits_{\sfera^2}\mathbf{K}^N_0[\psi]r^2dvdrd\omega+\int\limits_{v_1}^{v_2}\int\limits_{r_+}^{r_{\text{RED}}}\int\limits_{\sfera^2}C_1\cdot(r\psi)^2r^2dvdrd\omega \geq \\ \int\limits_{v_1}^{v_2}\int\limits_{r_+}^{r_{\text{RED}}}\int\limits_{\sfera^2}\left[(T(r\psi))^2+(\de_{\tilde{r}}(r\psi))^2+\abs{\nablas (r\psi)}^2+(r\psi)^2\right]r^2dvdrd\omega.
\end{multline}
Using also property \eqref{KawayH} for $-\mathbf{K}_0^N[\psi]$ in the region $[r_{\text{RED}},r_{\text{CUT}}]$, we arrive at:
\begin{align}
\int\limits_{v_1}^{v_2}\int\limits_{r_+}^{r_{\text{RED}}}\int\limits_{\sfera^2}\left[(T(r\psi))^2+(\de_{\tilde{r}}(r\psi))^2+\abs{\nablas (r\psi)}^2+(r\psi)^2\right]r^2dvdrd\omega+\overline{\mathbb{E}}[\psi](v_2) \lesssim \nonumber \\ \int\limits_{v_1}^{v_2}\int\limits_{r_{\text{RED}}}^{r_{\text{CUT}}}\int\limits_{\sfera^2}\left[(T(r\psi))^2+(\de_{\tilde{r}}(r\psi))^2+\abs{\nablas (r\psi)}^2+(r\psi)^2\right]r^2dvdrd\omega +\overline{\mathbb{E}}[\psi](v_1)+ \nonumber \\
+ \epsilon\int\limits_{v_1}^{v_2}\int\limits_{r_+}^{r_{\text{CUT}}}\int\limits_{\sfera^2}(\de_{\tilde{r}}(r\psi))^2dvdrd\omega+\Big(\frac{1}{\epsilon}+\frac{1}{C_0}+C_1\Big)\int\limits_{v_1}^{v_2}\int\limits_{r_+}^{r_{\text{CUT}}}\int\limits_{\sfera^2}(r\psi)^2dvdrd\omega.
\end{align}
The proof is concluded if we find a way to control all terms appearing in the last line. First of all, observe that any contribution made in the region $[r_{\text{RED}},r_{\text{CUT}}]$ can be included into the spacetime term in the second line, i.e. the one we obtained applying property \eqref{KawayH} to $\mathbf{K}^N_0[\psi]$. We are thus left with the contributions in the region $[r_+,r_{\text{RED}}]$:
\begin{align}
\int\limits_{v_1}^{v_2}\int\limits_{r_+}^{r_{\text{RED}}}\int\limits_{\sfera^2}\left[(T(r\psi))^2+(\de_{\tilde{r}}(r\psi))^2+\abs{\nablas (r\psi)}^2+(r\psi)^2\right]r^2dvdrd\omega+\overline{\mathbb{E}}[\psi](v_2) \lesssim \nonumber \\ \int\limits_{v_1}^{v_2}\int\limits_{r_{\text{RED}}}^{r_{\text{CUT}}}\int\limits_{\sfera^2}\left[(T(r\psi))^2+(\de_{\tilde{r}}(r\psi))^2+\abs{\nablas (r\psi)}^2+(r\psi)^2\right]r^2dvdrd\omega +\overline{\mathbb{E}}[\psi](v_1)+ \nonumber \\
+ \epsilon\int\limits_{v_1}^{v_2}\int\limits_{r_+}^{r_{\text{RED}}}\int\limits_{\sfera^2}(\de_{\tilde{r}}(r\psi))^2dvdrd\omega+\Big(\frac{1}{\epsilon}+\frac{1}{C_0}+C_1\Big)\int\limits_{v_1}^{v_2}\int\limits_{r_+}^{r_{\text{RED}}}\int\limits_{\sfera^2}(r\psi)^2dvdrd\omega.
\end{align}
We now apply Lemma \ref{lem:0-near-away} to the last integral and obtain:
\begin{align}
\int\limits_{v_1}^{v_2}\int\limits_{r_+}^{r_{\text{RED}}}\int\limits_{\sfera^2}\left[(T(r\psi))^2+(\de_{\tilde{r}}(r\psi))^2+\abs{\nablas (r\psi)}^2+(r\psi)^2\right]dvol_d+\overline{\mathbb{E}}[\psi](v_2) \lesssim \nonumber \\ \int\limits_{v_1}^{v_2}\int\limits_{r_{\text{RED}}}^{r_{\text{CUT}}}\int\limits_{\sfera^2}\left[(T(r\psi))^2+(\de_{\tilde{r}}(r\psi))^2+\abs{\nablas (r\psi)}^2+(r\psi)^2\right]r^2dvdrd\omega +\overline{\mathbb{E}}[\psi](v_1)+ \nonumber \\
+ \Big(\epsilon+\eta\Big(\frac{1}{\epsilon}+\frac{1}{C_0}+C_1\Big)\Big)\int\limits_{v_1}^{v_2}\int\limits_{r_+}^{r_{\text{RED}}}\int\limits_{\sfera^2}(\de_{\tilde{r}}(r\psi))^2dvdrd\omega+ \nonumber \\ 
+\frac{1}{\eta}\Big(\frac{1}{\epsilon}+\frac{1}{C_0}+C_1\Big)\int\limits_{v_1}^{v_2}\int\limits_{r_{\text{RED}}}^{r_{\text{CUT}}}\int\limits_{\sfera^2}\Big((\de_{\tilde{r}}(r\psi))^2+(r\psi)^2\Big)dvdrd\omega.
\end{align}
We choose $\epsilon$ and $\eta$ small enough that the contribution of $(\de_{\tilde{r}}(r\psi))^2$ in $[r_+,r_{\text{RED}}]$ can be absorbed into the spacetime term at the left hand side -- first choose $\epsilon\ll 1$ and then $\eta \ll \epsilon$. As a consequence, the coefficient in front of the last integral becomes very large, but this is not an issue since its contribution can be included in the spacetime term to the right of the ``$\lesssim$" sign. This concludes the proof.
\end{proof}

\subsection{Non-degenerate energy boundedness}\label{subsec:Nenergy}
In this subsection, we use the redshift effect (Proposition \ref{prop:redshift}) to establish a non-degenerate energy boundedness statement, which, together with Proposition \ref{prop:Tenergy}, completes the proof of Theorem \ref{thm:boundednessintro}. 
\begin{prop}\label{prop:Nenergy}
Let $\psi$ be as in Proposition \ref{prop:redshift}. Then, for any $v_0\leq v_1\leq v_2$, we have
	\[ \overline{\mathbb{E}}[\psi](v_2)\lesssim \overline{\mathbb{E}}[\psi](v_1).\]
\end{prop}
\begin{proof}
We know from Proposition \ref{prop:redshift} that
\begin{align}
    \int\limits_{v_1}^{v_2}\int\limits_{r_+}^{r_{\text{RED}}}\int\limits_{\sfera^2}\Big[(T(r\psi))^2+(\de_{\tilde{r}}(r\psi))^2+\abs{\nablas(r\psi)}^2+(r\psi)^2 \Big]dvdrd\omega +\overline{\mathbb{E}}[\psi](v_2) \lesssim \nonumber \\  \int\limits_{v_1}^{v_2}\int\limits_{{r_{\text{RED}}}}^{{r_{\text{CUT}}}}\int\limits_{\sfera^2}\Big[(T(r\psi))^2+(\de_{\tilde{r}}(r\psi))^2+\abs{\nablas(r\psi)}^2+(r\psi)^2\Big]dvdrd\omega + \overline{\mathbb{E}}[\psi](v_1).
\end{align}
Using Remark \ref{rem:controllingTrpsi}, we may write
\begin{multline}
\int\limits_{v_1}^{v_2}\int\limits_{{r_{\text{RED}}}}^{{r_{\text{CUT}}}}\int\limits_{\sfera^2}\Big[(T(r\psi))^2+ (\de_{\tilde{r}}(r\psi))^2+\abs{\nablas(r\psi)}^2+(r\psi)^2\Big]dvdrd\omega \lesssim \\  \int_{v_1}^{v_2}\mathbb{E}^T[\psi](v)dv \lesssim (v_2-v_1+\frac{\pi}{2})\mathbb{E}^T[\psi](v_1).
\end{multline}
Therefore
\begin{equation}\label{eq:pre-boundedness}
\int_{v_1}^{v_2}\overline{\mathbb{E}}[\psi](v)dv +\overline{\mathbb{E}}[\psi](v_1) \lesssim (v_2-v_1)\mathbb{E}^T[\psi](v_0)+\overline{\mathbb{E}}[\psi](v_0).
\end{equation}
A standard argument now yields $\overline{\mathbb{E}}[\psi](v_1)\lesssim \overline{\mathbb{E}}[\psi](v_0)$.
\end{proof}

\begin{remark}
It is possible to prove Proposition \ref{prop:Nenergy} without resorting to Proposition \ref{prop:redshift}. In particular, once we arrive at
\begin{align}\int\limits_{v_1}^{v_2}\int\limits_{r_+}^{r_{\text{RED}}}\int\limits_{\sfera^2}\left[(T(r\psi))^2+(\de_{\tilde{r}}(r\psi))^2+\abs{\nablas (r\psi)}^2+(r\psi)^2\right]dvol_d+\overline{\mathbb{E}}[\psi](v_1) \lesssim \nonumber \\ \int\limits_{v_1}^{v_2}\int\limits_{r_{\text{RED}}}^{r_{\text{CUT}}}\int\limits_{\sfera^2}\left[(T(r\psi))^2+(\de_{\tilde{r}}(r\psi))^2+\abs{\nablas (r\psi)}^2+(r\psi)^2\right]r^2dvdrd\omega +\overline{\mathbb{E}}[\psi](v_1)+ \nonumber \\
+\int\limits_{v_1}^{v_2}\int\limits_{r_+}^{r_{\text{RED}}}\int\limits_{\sfera^2}(r\psi)^2dvdrd\omega,
\end{align}
we can control the zeroth order term with $\int_{v_1}^{v_2}\mathbb{E}^T[\psi](v)dv$ times a big enough constant $\tilde{C}$, since the energy $\mathbb{E}^T[\psi](v)$ has a zeroth order term, instead of using Lemma \ref{lem:0-near-away}. In so doing, we would arrive directly at eq. \eqref{eq:pre-boundedness} in the proof of Theorem \ref{prop:Nenergy}.
\end{remark}
\section{Integrated decay} \label{sec:decay}
In this section, we retrace the steps described in Subsection \ref{subsec:sketch} to arrive at the main result of this paper, namely the integrated decay estimate for $\overline{\mathbb{E}}$ (Theorem \ref{thm:integrateddecayintro}). Firstly, we prove Theorem \ref{thm:deg-integrated} in Subsection \ref{subsec:degdecay}, containing a degenerate (at $\Hcal^+$, at the photon sphere and at $\scri^+$) Morawetz estimate, which is conditional on $M^2/l^2\leq 2/27$ for $\psi$ and unconditional for $T\psi$. The ``conditional" arises from the way the boundary term at $\scri^+$ in the divergence identity is treated. The estimate for $T\psi$ guarantees non-degenerate (at $\Hcal^+$, at the photon sphere and at $\scri^+$) control of the $T$ derivative. To prove Theorem \ref{thm:integrateddecayintro}, we extend control from $T(r\psi)$ to the remaining spatial derivatives and to $r\psi$ itself using an elliptic estimate on the elliptic part of the $\square_g$ operator. This is the content of Subsection \ref{subsec:elliptic} and concludes the proof of Theorem \ref{thm:integrateddecayintro}. Finally, in Subsection \ref{subsec:polydecay}, we promote the integrated decay estimate to inverse polynomial decay in time.
\subsection{The (degenerate) integrated decay estimate}\label{subsec:degdecay}
Consider the Morawetz multiplier
\begin{equation}\label{morawetzmultiplier}
	f(r)\cdot R^*(r\psi)+\dfrac{f'(r)}{2}\cdot (r\psi) \quad \text{with} \quad	f(r)=\Big(1-\dfrac{3M}{r}\Big)\sqrt{1+\dfrac{6M}{r}}.
\end{equation}
The choice of $f(r)$ is inspired by the one used in \cite{holzegel2024note} to deduce an integrated decay estimate in Schwarzschild-de Sitter and it has the following favourable properties.
\begin{lem}\label{lem:choiceoff}
	There exist positive real constants $C_1$ and $C_2$ that only depend on black hole parameters, such that, for all $r\in[r_+,\infty)$, we have
	\[ \left|f'(r)\frac{r^3}{1-\frac{2M}{r}+\frac{r^2}{l^2}}\right| \leq C_1, \quad \left|f''(r)\frac{r^2}{1-\frac{2M}{r}+\frac{r^2}{l^2}}\right| \leq C_2. \]
	In particular, we have that
	\[ \lim\limits_{r\to\infty}f''(r)=-\dfrac{27M^2}{l^4} \]
\end{lem}
\begin{proof}
This is a direct computation.
\end{proof}
\begin{thm}[Degenerate integrated decay estimate] \label{thm:deg-integrated}
Let $\psi$ be a smooth solution of \eqref{WE} arising from initial data as in Theorem \ref{thm:wp}, subject to dissipative boundary conditions \eqref{BC}. Let $v_0\leq v_1\leq v_2$ and consider the integrated energy
    \begin{equation*}
         \mathbb{I}_{\text{deg}}[\psi](v_1,v_2):= \int\limits_{v_1}^{v_2}\int\limits_{r_+}^{\infty}\int\limits_{\sfera^2}\left[\dfrac{(R^{*}(r\psi))^{2}}{r^3}+\dfrac{1}{r}\left(1-\dfrac{3M}{r}\right)^{2}\abs{\nablas(r\psi)}^{2}+\dfrac{(r\psi)^{2}}{r^2}\right]dvdrd\omega.
    \end{equation*}
\begin{enumerate}
    \item If the black hole parameters $M,l$ satisfy the condition $\frac{M^{2}}{l^{2}}\leq\frac{2}{27}$ and $\psi$ has a vanishing spherical average, then $\psi$ satisfies the integrated decay estimate:
    \begin{equation*}
    \mathbb{I}_{\text{deg}}[\psi](v_{0},v_{1}) \lesssim \overline{\mathbb{E}}[\psi](v_{0}).
    \end{equation*}
    \item The solution $T\psi$ satisfies the (unconditional) integrated decay estimate:
    \begin{equation*}
    \mathbb{I}_{\text{deg}}[T\psi](v_1,v_2) \lesssim \overline{\mathbb{E}}[\psi](v_1)+ \overline{\mathbb{E}}[T\psi](v_1).
    \end{equation*}
\end{enumerate}
\end{thm}
\begin{proof}[Proof of Theorem \ref{thm:deg-integrated}]
Multiply equation \eqref{WErpsi} by the Morawetz multiplier \eqref{morawetzmultiplier}, integrate by parts in $(1-\frac{2M}{r}+\frac{r^2}{l^2})^{-1}dvdrd\omega$ over $[v_1,v_2]\times\R\times\sfera^2$ and arrive at:
\begin{multline}
\int\limits_{v_1}^{v_2}\int\limits_{r_+}^{\infty}\int\limits_{\sfera^2}\left[f'(r)(R^{*}(r\psi))^{2}-\dfrac{1}{2}f(r)\left(\dfrac{1-\frac{2M}{r}+\frac{r^2}{l^2}}{r^{2}}\right)'\abs{r\nablas(r\psi)}^{2}+\right. \hfill \\
\hfill \left.-\left(\dfrac{1}{2}f(r)V'(r)+\frac{1}{4}f'''(r)\right)(r\psi)^{2}\right]\dfrac{dvdrd\omega}{(1-\frac{2M}{r}+\frac{r^2}{l^2})} = \\
= \int\limits_{\Sigma_{v}}\left[\dfrac{f(r)}{2}\left(\Big(1-\frac{2M}{r}+\frac{r^2}{l^2}\Big)(\de_{\tilde{r}}(r\psi))^{2}-\abs{\nablas(r\psi)}^{2}-\dfrac{2M}{r^3}(r\psi)^{2}\right)\right. \hfill \\
\hfill \left.-\dfrac{f'(r)}{4}\de_{\tilde{r}}(r\psi)^{2}-\dfrac{f''(r)}{4(1-\frac{2M}{r}+\frac{r^2}{l^2})}(r\psi)^{2}\right]drd\omega\,\Bigg\rvert_{v_1}^{v_2} \\
+\int\limits_{\scri^{+}(v_1,v_2)}\left[\dfrac{f(r)}{2}\left(2(T(r\psi))^{2}-\Big(1-\frac{2M}{r}+\frac{r^2}{l^2}\Big)\abs{\nablas(r\psi)}^{2}\right)+\dfrac{f'(r)}{4}R^{*}(r\psi)^{2}-\dfrac{f''(r)}{4}(r\psi)^{2}\right]dvd\omega \hfill \\
\hfill -\int\limits_{\mathcal{H}^{+}(v_1,v_2)}\left[f(r_{+})(T(r\psi))^{2}+\dfrac{f'(r)}{4}R^{*}(r\psi)^{2}-\dfrac{f''(r)}{4}(r\psi)^{2}\right]dv\, d\omega. \label{eq:integralid}
\end{multline}
We refer to the term on the left hand side of the equal sign as bulk term and denote it as follows:
\begin{multline}
\mathbb{I}_{\text{bulk}}[\psi](v_1,v_2):=\int\limits_{v_1}^{v_2}\int\limits_{r_+}^{\infty}\int\limits_{\sfera^2}\left[f'(r)(R^{*}(r\psi))^{2}-\dfrac{1}{2}f(r)\left(\dfrac{1-\frac{2M}{r}+\frac{r^2}{l^2}}{r^{2}}\right)'\abs{r\nablas(r\psi)}^{2}+\right. \hfill \\
\hfill \left.-\left(\dfrac{1}{2}f(r)V'(r)+\frac{1}{4}f'''(r)\right)(r\psi)^{2}\right]\dfrac{dvdrd\omega}{(1-\frac{2M}{r}+\frac{r^2}{l^2})}.
\end{multline}
We refer to the terms on the right of the equal sign as boundary terms. We will analyze bulk and boundary terms separately with the goal of establishing:
\begin{itemize}
    \item a \textbf{lower bound} for the \textbf{bulk term} of the type $\mathbb{I}_{\text{deg}}[\psi](v_1,v_2)\lesssim\mathbb{I}_{\text{bulk}}[\psi](v_1,v_2)$;
    \item an \textbf{upper bound} for the \textbf{boundary terms} using $\overline{\mathbb{E}}[\psi](v_1)+\overline{\mathbb{E}}[T\psi](v_1)$.
\end{itemize}

\textbf{Analysis of the boundary terms.} Using the properties of $f(r)$ stated in Lemma \ref{lem:choiceoff}, one sees that the flux through $\Sigma_v$ is controlled by $\overline{\mathbb{E}}[\psi](v)$:
\begin{multline*}
\int\limits_{\Sigma_{v}}\left[\dfrac{f(r)}{2}\left(\Big(1-\frac{2M}{r}+\frac{r^2}{l^2}\Big)(\de_{\tilde{r}}(r\psi))^{2}-\abs{\nablas(r\psi)}^{2}-\dfrac{2M}{r^3}(r\psi)^2\right)\right.+ \\ +\left.\dfrac{f'(r)}{4}\de_{\tilde{r}}(r\psi)^{2}-\dfrac{f''(r)}{4(1-\frac{2M}{r}+\frac{r^2}{l^2})}(r\psi)^{2}\right]drd\omega\lesssim \overline{\mathbb{E}}[\psi](v).
\end{multline*}
The inequality is immediate for every term except the mixed term coming from $\de_{\tilde{r}}(r\psi)^2$, which we control after application of Cauchy-Schwarz:
\begin{equation}
\int\limits_{\Sigma_v}\frac{f'(r)}{4}\de_{\tilde{r}}(r\psi)^2drd\omega\leq\int\limits_{\Sigma_v}\frac{f'(r)	}{4}\Big(r^2(\de_{\tilde{r}}(r\psi))^2+\frac{(r\psi)^2}{r^2}\Big)drd\omega\lesssim \overline{\mathbb{E}}[\psi].
\end{equation}

\noindent We control the flux along $\Hcal^+(v_1,v_2)$ using the energy conservation law for $\mathbb{E}^T[\psi]$:
\begin{equation}
0\leq -\int\limits_{\mathcal{H}^{+}(v_1,v_2)}f(r_{+})(T(r\psi))^{2}dvd\omega \lesssim \mathbb{E}^T[\psi](v_1)\lesssim\overline{\mathbb{E}}[\psi](v_1).
\end{equation}

\noindent Along the boundary at $\scri^{+}$, using again Lemma \ref{lem:choiceoff}, we find
\begin{equation}
\int\limits_{\scri^{+}(v_1,v_2)}\left[(T(r\psi))^{2}-\frac{r^2}{2l^2}\abs{\nablas(r\psi)}^{2}+\dfrac{27M^2}{4l^4}(r\psi)^{2}\right]dvd\omega.
\end{equation}
We control $(T(r\psi))^{2}$ using the energy conservation law for $\mathbb{E}^T[\psi]$. We could move the angular derivatives to the left hand side and obtain control but we would not know what to do with the zeroth order term. Therefore, we absorb it into the angular term by way of the Poincaré-type inequality:
\begin{equation}\label{poincare}
    \int_{\sfera^{2}}|\nabla_{\sfera^{2}}(\psi-\psi_0)|^{2}\,d\omega 
    \;\geq\; 2\int_{\sfera^{2}}|\psi-\psi_0|^{2}\,d\omega,
\end{equation}
where $\psi_0$ is the spherical average of $\psi$. Then, under the additional hypothesis that $\psi_0=0$, we obtain:
\begin{equation}
\int\limits_{\scri^{+}(v_1,v_2)}\left[-\frac{r^2}{2l^2}\abs{\nablas(r\psi)}^{2}+\dfrac{27M^2}{4l^4}(r\psi)^{2}\right]dvd\omega \leq\int\limits_{\scri^{+}(v_1,v_2)}\dfrac{1}{2}\left(-\dfrac{1}{l^{2}}+\dfrac{27M^{2}}{2l^{4}}\right)\abs{r\nablas(r\psi)}^{2}dvd\omega.
\end{equation}
If $\frac{M^{2}}{l^{2}}<\frac{2}{27}$, the final coefficient in front of the angular derivatives is negative, so we may move the integral to the left hand side and obtain control.
\vspace{3mm}

\textbf{Analysis of the bulk term.} The lower bound $\mathbb{I}_{\text{deg}}[\psi](v_1,v_2)\lesssim\mathbb{I}_{\text{bulk}}[\psi](v_1,v_2)$ is immediate for the first order terms but non-trivial for the zeroth order term. Precisely, since $f(r)$ is monotonically increasing with a change in sign at $r=3M$, and it behaves asymptotically according to Lemma \ref{lem:choiceoff}, the first order terms are manifestly non-negative. On the contrary, establishing control of the zeroth order requires application of a Hardy-type inequality; due to the technicality of this procedure, we complete the proof with all its details in the Appendix.

This concludes the proof of the first statement of Theorem \ref{thm:deg-integrated}, i.e. the degenerate and conditional Morawetz estimate for $r\psi$.

To obtain statement 2, i.e. the unconditional, degenerate Morawetz estimate for $T\psi$, apply the divergence identity \eqref{eq:integralid} to the solution $T\psi$. Repeat the proof of statement 1 using $T\psi$ instead of $\psi$, but, when dealing with the boundary term at $\scri^+$, observe that it is no longer necessary to invoke Poincaré inequality \eqref{poincare} because one finds
\begin{equation}
\int\limits_{\scri^{+}(v_1,v_2)}\left[(TT(r\psi))^{2}-\frac{r^2}{2l^2}\abs{\nablas T(r\psi)}^{2}+\dfrac{27M^2}{4l^4}(T(r\psi))^{2}\right]dvd\omega,
\end{equation}
which is controlled by $\mathbb{E}^T[\psi](v_1)+\mathbb{E}^T[T\psi](v_1)$ without restrictions on $M,l$ or on the spherical average. This concludes the proof of statement 2 of Theorem \ref{thm:deg-integrated}.
\end{proof}

\subsection{Completing the proof of Theorem \ref{thm:integrateddecayintro}}\label{subsec:elliptic}
\begin{proof}[Proof of Theorem \ref{thm:integrateddecayintro}]
From statement $2$ of Theorem \ref{thm:deg-integrated}, we have non-degenerate control of $T\psi$, since it implies that
\begin{equation}\label{eq:Tcontrol}
\int\limits_{v_1}^{v_2}\int\limits_{r_+}^{\infty}\int\limits_{\sfera^2}\dfrac{(T(r\psi))^{2}}{r^{2}}dvdrd\omega\lesssim \overline{\mathbb{E}}[\psi](v_1)+\overline{\mathbb{E}}[T\psi](v_1).
\end{equation}
To control the remaining derivatives and $r\psi$ itself, we prove an elliptic estimate for the elliptic component of $\square_g$. Multiply eq. \eqref{WErpsi} by $h(r)\cdot (r\psi)$, integrate by parts over $[v_1,v_2]\times[r_+,\infty)\times\sfera^2$ in $(1-\frac{2M}{r}+\frac{r^2}{l^2})^{-1}dvdrd\omega$ and arrive at:
\begin{multline}
\int\limits_{v_1}^{v_2}\int\limits_{r_+}^{\infty}\int\limits_{\sfera^2}\left[h(r)(R^{*}(r\psi))^{2}+h(r)\Big(1-\frac{2M}{r}+\frac{r^2}{l^2}\Big)\abs{\nablas(r\psi)}^{2}+\right. \hfill \\
\hfill \left.+\left(h(r)V(r)-\dfrac{h''(r)}{2}\right)(r\psi)^{2}\right]\dfrac{dvdrd\omega}{(1-\frac{2M}{r}+\frac{r^2}{l^2})}= \\
\int\limits_{v_1}^{v_2}\int\limits_{r_+}^{\infty}\int\limits_{\sfera^2}\dfrac{h(r)(T(r\psi))^{2}}{(1-\frac{2M}{r}+\frac{r^2}{l^2})}dvdrd\omega+\left.\int\limits_{\Sigma_{v}}\left[-\dfrac{h(r)}{2}\de_{\tilde{r}}(r\psi)^{2}-\dfrac{h'(r)}{2(1-\frac{2M}{r}+\frac{r^2}{l^2})}(r\psi)^{2}\right]drd\omega\right\rvert_{v_1}^{v_2}+ \\
\int\limits_{\scri^{+}(v_1,v_2)}\left[\dfrac{h(r)}{2}R^{*}(r\psi)^{2}-\dfrac{h'(r)}{2}(r\psi)^{2}\right]dvd\omega -\int\limits_{\mathcal{H}^{+}(v_1,v_2)}\left[\dfrac{h(r)}{2}R^{*}(r\psi)^{2}-\dfrac{h'(r)}{2}(r\psi)^{2 }\right]dvd\omega, \label{eq:elliptic-ibp}
\end{multline}
We refer to the left-hand side as the bulk term. The spacetime integral involving $(T(r\psi))^2$ is placed on the right because, for a suitable $h(r)$, it is controlled by $\overline{\mathbb{E}}[\psi](v_1)+\overline{\mathbb{E}}[T\psi](v_1)$ via Theorem~\ref{thm:deg-integrated}. Precisely, this is achieved if $h(r)$ vanishes near $\Hcal^+$ at the rate of $(1-\frac{2M}{r}+\frac{r^2}{l^2})$, which rules out the naive choice $h\equiv 1$.

\vspace{3mm}
\noindent\textbf{Choice of $h(r)$ and resulting difficulties.} Subject to the above constraint, we require additionally that $h\geq 0$ for positivity of the bulk integrand, and that $\de_rh\to 0$ as $r\to\infty$ to control the integral along $\scri^+$. The simplest admissible choice is a smooth cut-off that rises monotonically from $0$ to $1$ over a finite interval $[r_0,r_1]\subset(r_+,\infty)$. This choice, however, introduces two difficulties:
\begin{enumerate}
    \item Non-degenerate control is lost for $r\leq r_1$. This is resolved by adding a small multiple of the redshift estimate (Proposition~\ref{prop:redshift}), which recovers control near $\Hcal^+$ at the cost of having control away from it. Naturally, this is only applicable if the transition region of $h(r)$ ends before $r_{\mathrm{RED}}$, so we set $r_1=r_{\mathrm{RED}}$.
    \item In the transition region, $-h''(r)$ is not sign-definite, producing a potentially negative bulk contribution. This is handled by integrating by parts in $r$ and reducing the problem to controlling a term proportional to $h'(r)$, which can be made arbitrarily small by spreading the transition of $h$ toward $r^*\to-\infty$.
\end{enumerate}
We now verify that the boundary terms are controlled by the energy $\overline{\mathbb{E}}[\psi]$ and then implement the above strategy.
\vspace{3mm}

\noindent \textbf{Boundary terms analysis.} Since $h(r)$ vanishes at $\Hcal^+$, we have no contribution on $\Hcal^+$. Along $\Sigma_v$, we integrate by parts and find
\begin{align}
    \int\limits_{\Sigma_v}\left[-\frac{h(r)}{2}\de_{\tilde{r}}(r\psi)^2-\frac{h'(r)}{2\Big(1-\frac{2M}{r}+\frac{r^2}{l^2}\Big)}(r\psi)^2\right]drd\omega\Bigg|_{v_0}^{v_1}=\int\limits_{\scri^+(v)}-\frac{1}{2}(r\psi)^2d\omega\Bigg|_{v_1}^{v_2}.
\end{align}
Along $\scri^+(v_1,v_2)$, using boundary conditions, we find
\begin{align}
    \int\limits_{\scri^+(v_1,v_2)}\frac{1}{2}R^*(r\psi)^2dvd\omega=\int\limits_{\scri^+(v_1,v_2)}-\frac{1}{2}T(r\psi)^2dvd\omega=\left.\int\limits_{\scri^+(v)}-\frac{1}{2}(r\psi)^2d\omega\right\rvert_{v_1}^{v_2}.
\end{align}
The integrals on the spheres at infinity do not cancel out, but we control their sum with the energy flux along the relative constant $v$ hypersurface:
\begin{align}
    \int\limits_{\scri^+(v)}(r\psi)^2d\omega=\int\limits_{\Sigma_v}\de_{\tilde{r}}(\zeta(r\psi))^2drd\omega \leq \int\limits_{\Sigma_v}\Big((\de_{\tilde{r}}\zeta+\frac{\zeta}{r^2})(r\psi)^2+\zeta r^2(\de_{\tilde{r}}(r\psi))^2\Big)drd\omega \lesssim \mathbb{E}^T[\psi](v),   
\end{align}
where $\zeta$ is a smooth cut-off function that monotonically increases from $0$ to $1$ somewhere away from the event horizon. Once the boundary terms are taken care of, we arrive at:
\begin{multline}\label{eq:elliptic-bulk}
\int\limits_{v_1}^{v_2}\int\limits_{r_+}^{\infty}\int\limits_{\sfera^2}h(r)\left[\frac{(R^{*}(r\psi))^{2}}{(1-\frac{2M}{r}+\frac{r^2}{l^2})}+\abs{\nablas(r\psi)}^2+\frac{2M}{r^3}(r\psi)^2\right]dvdrd\omega + \\
\int\limits_{v_1}^{v_2}\int\limits_{r_+}^{\infty}\int\limits_{\sfera^2}-\dfrac{h''(r)(r\psi)^2}{2(1-\frac{2M}{r}+\frac{r^2}{l^2})}dvdrd\omega \lesssim \mathbb{E}^T[\psi](v_1)+\mathbb{E}^T[T\psi](v_1).
\end{multline}
\textbf{Bulk term analysis.} We now have to address the lack of control for $r\leq r_{\text{RED}}$ due of $h(r)$ vanishing near $\Hcal^+$ and the negative contribution given by $-h''(r)$ in the transition region $[r_0,r_{\text{RED}}]$. We implement the strategy described above in points 1 and 2.
\vspace{3mm}

\noindent \textbf{Step 1.} We add a $\delta$ amount of eq. \eqref{redshift} to eq. \eqref{eq:elliptic-bulk} and obtain:
\begin{align}
\delta\int\limits_{v_0}^{v_1}\int\limits_{r_+}^{r_{\text{RED}}}\int\limits_{\sfera^2}\Big[(T(r\psi))^2+(\de_{\tilde{r}}(r\psi))^2+\abs{\nablas(r\psi)}^2+\frac{(r\psi)^2}{r^2} \Big]dvdrd\omega + \nonumber \\
\int\limits_{v_0}^{v_1}\int\limits_{r_{\text{RED}}}^{\infty}\int\limits_{\sfera^2}\left[(R^{*}(r\psi))^{2}+\abs{\nablas(r\psi)}^{2}+h(r)V(r)(r\psi)^{2}\right]dvdrd\omega \lesssim \nonumber \\
\delta\int\limits_{v_0}^{v_1}\int\limits_{{r_{\text{RED}}}}^{{r_{\text{CUT}}}}\int\limits_{\sfera^2}\Big[(T(r\psi))^2+(\de_{\tilde{r}}(r\psi))^2+\abs{\nablas(r\psi)}^2+(r\psi)^2\Big]dvdrd\omega +\mathbb{E}^N [\psi](v_0) + \nonumber \\
\int\limits_{v_0}^{v_1}\int\limits_{r_+}^{\infty}\int\limits_{\sfera^2}\dfrac{h''(r)(r\psi)^2}{2(1-\frac{2M}{r}+\frac{r^2}{l^2})}dvdrd\omega + \mathbb{E}^T[\psi](v_0)+\mathbb{E}^T[T\psi](v_0).
\end{align}
Issue 1 is resolved, since we now have control before $r_{\text{RED}}$; the price to pay is the new spacetime integral on the right hand side, whose contribution is on the region $[r_{\text{RED}},r_{\text{CUT}}]$. To take care of it, we choose $\delta$ small enough to absorb it on the left hand side, into the spacetime integral whose contribution is on $[r_{\text{RED}},\infty)$. We thus arrive at
\begin{multline}
\int\limits_{v_0}^{v_1}\int\limits_{r_+}^{\infty}\int\limits_{\sfera^2}\left[\dfrac{(T(r\psi))^2}{r^2}+r^2(\de_{\tilde{r}}(r\psi))^2+\abs{\nablas(r\psi)}^2+\frac{(r\psi)^2}{r^2}\right]\lesssim \mathbb{E}^N[\psi](v_0)+\mathbb{E}^N[T\psi](v_0) + \\
+\int\limits_{v_0}^{v_1}\int\limits_{r_+}^{\infty}\int\limits_{\sfera^2}\dfrac{h''(r)}{2(1-\frac{2M}{r}+\frac{r^2}{l^2})}(r\psi)^2dvdrd\omega,
\end{multline}
where we implicitly improved the weight of $(r\psi)^2$ at infinity in the spacetime term (see Lemma \ref{lem:weightimprovement}) and used the fact that, since we control $\frac{(T(r\psi))^2}{r^2}$ from eq. (\ref{eq:Tcontrol}), controlling $(R^*(r\psi))^2/(1-\frac{2M}{r}+\frac{r^2}{l^2})$ is equivalent to controlling $r^2(\de_{\tilde{r}}(r\psi))^2$.
\vspace{3mm}

\noindent\textbf{Step 2.} To control the spacetime integral featuring $h''(r)$, we integrate it by parts once and apply Cauchy-Schwarz:
\begin{multline}
\int\limits_{v_{0}}^{v_{1}}\int\limits_{r_+}^{\infty}\int\limits_{\sfera^{2}}\dfrac{h''(r)}{2(1-\frac{2M}{r}+\frac{r^2}{l^2})}(r\psi)^{2}dvdrd\omega=\underbrace{\left.\int\limits_{v_0}^{v_1}\int\limits_{\sfera^2}\dfrac{h'(r)}{2}(r\psi)^{2}dvd\omega\right\rvert_{r_+}^{\infty}}_{=0}-\int\limits_{v_{0}}^{v_{1}}\int\limits_{r_+}^{\infty}\int\limits_{\sfera^{2}}h'(r)\de_{\tilde{r}}(r\psi)(r\psi)dvdrd\omega \\
\leq \int\limits_{v_0}^{v_1}\int\limits_{r_+}^{\infty}\int\limits_{\sfera^2}h'(r)\Big((\de_{\tilde{r}}(r\psi))^2+(r\psi)^2 \Big)dvdrd\omega.
\end{multline}
We claim that $h'(r)$ can be made arbitrarily small while maintaining the desired profile for $h(r)$. In fact, with respect to the tortoise coordinate $r^*$, the smooth cut-off function $h$ transitions from $0$ to $1$ between two points $r^*_0$ and $r^*_1=r^*_{\text{RED}}$, corresponding respectively to $r_0$ and $r_{\text{RED}}$. Since the event horizon $r_+$ corresponds to $r^*(r_+)= -\infty$, we have infinite space to operate the transition from $0$ to $1$, i.e. we can push $r^*_0$ towards $-\infty$ as much as needed and, in so doing, make $h'(r)$ as small as necessary to absorb the error on the left hand side.
\end{proof}

\subsection{Arbitrarily fast polynomial decay}\label{subsec:polydecay}
To deduce arbitrarily fast polynomial decay, we use a rather standard argument, often found in the literature. We propose here the following formulation.
\begin{lem}\label{lem:polydecay}
    Let $\psi$ be a solution to the conformal wave equation \eqref{WE}, and let $\mathcal{E}[\psi](v)$ be a positive energy depending smoothly on $\psi$ and its derivatives near $v$. Assume:
\begin{enumerate}
    \item $\mathcal{E}[\psi](v)$ is $C^1$ in $v$;
    \item $\mathcal{E}[\psi](v_2)\lesssim \mathcal{E}[\psi](v_1)$ for all $v_0\leq v_1\leq v_2$;
    \item $\displaystyle\int_{v_1}^{v_2}\mathcal{E}[\psi](v)\,dv \lesssim \mathcal{E}[\psi](v_1)+\mathcal{E}[T\psi](v_1)$ for all $v_0\leq v_1\leq v_2$.
\end{enumerate}
    Then there exists a constant $C_n$, depending only on $n$, such that 
    \[\mathcal{E}[\psi](v)\lesssim \dfrac{C_n}{(1+v)^n}\sum_{k=0}^n\mathcal{E}[T^k\psi](v_0).\]
\end{lem}
\begin{proof}
Since $T$ commutes with \eqref{WE}, hypotheses 2--3 also hold for $T^k\psi$ for all $k\leq n$. Fix $v>v_0$. Applying assumption 3 on $[v/2,v]$ and the mean value theorem (using assumption 1) gives an $s\in[v/2,v]$ with $\mathcal{E}[\psi](s)\lesssim \frac{1}{v}(\mathcal{E}[\psi](v/2)+\mathcal{E}[T\psi](v/2))$. Propagating from $s$ to $v$ via assumption 2 yields
\begin{equation}\label{eq:basic}
\mathcal{E}[\psi](v) \;\lesssim\; \frac{1}{v}\bigl(\mathcal{E}[\psi](v/2)+\mathcal{E}[T\psi](v/2)\bigr).
\end{equation}
Applying the same argument to $\mathcal{E}[T\psi](v/2)$ and substituting back into \eqref{eq:basic} gives $\mathcal{E}[\psi](v)\lesssim \frac{1}{v^2}\sum_{k=0}^{2}\mathcal{E}[T^k\psi](v/4)$ -- here we implicitly reduce to the only relevant case of $v>1$. Iterating $n$ times:
\begin{equation}\label{eq:iterate}
\mathcal{E}[\psi](v) \;\lesssim\; \frac{C_n}{v^n}\sum_{k=0}^n \mathcal{E}[T^k\psi]\!\left(\tfrac{v}{2^n}\right).
\end{equation}
The conclusion follows by applying boundedness assumption 2 to replace $\mathcal{E}[T^k\psi](v/2^n)$ by $\mathcal{E}[T^k\psi](v_0)$ for each $k\leq n$.
\end{proof}
\begin{proof}[Proof of Corollary \ref{cor:decay}]
The energy $\overline{\mathbb{E}}[\psi](v)$ satisfies all hypotheses present in Lemma \ref{lem:polydecay}, hence we apply it and obtain the desired result.
\end{proof}
\appendix
\section{A lower bound for \texorpdfstring{$\mathbb{I}_{\text{bulk}}[\psi]$}{I bulk}}
The goal of this Appendix is to conclude the proof of the lower bound
\begin{equation}\label{lowerbound}
\mathbb{I}_{\text{deg}}[\psi](v_1,v_2)\lesssim\mathbb{I}_{\text{bulk}}[\psi](v_1,v_2),
\end{equation}
which is a fundamental element in the proof of Theorem \ref{thm:deg-integrated}. With our choice of $f(r)$, the lower bound is immediate for the first order terms but non-trivial because the zeroth order coefficient 
\[ c_0(r,M,l):=-\frac{fV'}{2}-\frac{f'''}{4}\]
is negative near the horizon $\Hcal^+$. This is easily observed by writing
\begin{align*}
c_0(r,M,l):=&\Big(1-\frac{2M}{r}+\frac{r^2}{l^2}\Big)\left[-\dfrac{f(r)}{2}\dfrac{2M}{r^{3}}\de_{r}\Big(1-\frac{2M}{r}+\frac{r^2}{l^2}\Big)-\dfrac{\de_{r}f(r)}{4}\Big(\de_{r}\Big(1-\frac{2M}{r}+\frac{r^2}{l^2}\Big)\Big)^{2}\right] \\
&+o\Big(\Big(1-\frac{2M}{r}+\frac{r^2}{l^2}\Big)^2\Big),
\end{align*}
and checking that, in the Schwarzschild limit $l^2\rightarrow\infty$, the leading order coefficient equals
\[ \frac{1}{4} \left(\frac{3 M \left(1-\frac{3 M}{r}\right)}{r^2 \sqrt{\frac{6 M}{r}+1}}-\frac{3 M \sqrt{\frac{6 M}{r}+1}}{r^2}\right) \left(\frac{2 M}{r^2}\right)^2+\frac{M (2 M) \left(\frac{3 M}{r}-1\right) \sqrt{\frac{6 M}{r}+1}}{r^3 r^2},\]
which, at the event horizon $r=2M$, takes the negative value of $-\frac{11}{256 M^3}$. To compensate for this negativity, we borrow positivity from the $(R^*(r\psi))^2$ term using a Hardy-type inequality.

\subsection{Application of the Hardy-type inequality}\label{app:hardy}
Observe that, for any $g(r)\in C^1([r_+,\infty))$ that decays at least as fast as $r^{-1}$ at infinity, we have
\begin{multline}
    \int\limits_{[v_{0},v_{1}]\times(-\infty,\frac{\pi}{2}]\times\sfera^{2}}\left[f'(r)(R^{*}(r\psi))^{2}-\left(\dfrac{1}{2}f(r)V'(r)+\dfrac{f'''(r)}{4}\right)(r\psi)^{2}\right]
    dv\,dr^{*}\,d\omega \\
    =\int\limits_{[v_{0},v_{1}]\times(-\infty,\frac{\pi}{2}]\times\sfera^{2}}\left[f'(r)
    \abs{R^{*}(r\psi)+g(r)(r\psi)}^{2} +H(r,M,l)\,(r\psi)^{2}\right]dv\,dr^{*}\,d\omega \\
    -\int\limits_{[v_{0},v_{1}]\times(-\infty,\frac{\pi}{2}]\times\sfera^{2}}
    \Big(\de_{v}+\Big(1-\frac{2M}{r}+\frac{r^2}{l^2}\Big)\de_{\tilde{r}}\Big)\!\left[f'(r)g(r)(r\psi)^{2}\right]dv\,dr^{*}\,d\omega,
\end{multline}
where we define the zeroth order coefficient
\begin{equation}
    H(r,M,l) \;:=\; -\dfrac{f(r)V'(r)}{2}-\dfrac{f'''(r)}{4}
    +(f'(r)g(r))'-f'(r)(g(r))^{2}.
\end{equation}
If we can find $g(r)$ such that $H_0(r,M,l)\geq0$ for all $r\geq r_+(M,l)$, the lower bound \eqref{lowerbound} is guaranteed. The extra boundary terms are controlled by $\overline{\mathbb{E}}[\psi]$ because $g(r)$ is bounded at $\Hcal^+$ and decays at least as fast as $r^{-1}$ at $\scri^{+}$.

\subsection{Choice of \texorpdfstring{$g(r)$}{g(r)}}
After numerical inspection, we set
\begin{equation}
    g(r) \;=\; -\frac{M}{2\beta l^2}+\frac{M^2}{\beta r^3}+\frac{M}{r^2}-\frac{1}{2r},
\end{equation}
where $\beta$ will depend on the ratio $k=M/l$. To reduce the number of free parameters from two to one, we introduce the dimensionless variable $x=r/2M$ and the parameter $k=M/l$. The event horizon is then located at $x_+(k)$, whose fundamental property is that $x_+(0)=1$ (in the Schwarzschild limit $l\rightarrow\infty$) and that $x_+(k)$ monotonically decreases to $0$ for $k\rightarrow\infty$. We compute the zeroth order coefficient with respect to $(x,k)$ and find
\begin{equation}\label{Hbeta}
    H_{\beta}(2Mx,M,M/k) \;=\; 
    \frac{4k^2x^3+x-1}{512\beta^2 M^3 x^{12}\!\left(\frac{x+3}{x}\right)^{\!5/2}}
    \cdot h_{\beta}(x,k),
\end{equation}
where
\begin{multline}
    h_{\beta}(x,k) \;=\; -243+x(-162+5346\beta)+x^2(-27-648\beta-3807\beta^2) \\
    +x^3(1944k^2-2430\beta+1566\beta^2)
    +x^4\bigl(1296k^2+(-540-23328k^2)\beta+1485\beta^2\bigr) \\
    +x^5\bigl(216k^2-11016k^2\beta+(-96+14256k^2)\beta^2\bigr) \\
    +x^6\bigl(-3888k^4+1512k^2\beta+(52+8640k^2)\beta^2\bigr) \\
    +x^7\bigl(-2592k^4+(864+7776k^2)k^2\beta+(96+432k^2)\beta^2\bigr) \\
    +x^8\bigl(-432k^4+7776k^4\beta+(-768+3888k^2)k^2\beta^2\bigr) \\
    +x^9 k^2\bigl(1728k^2\beta+(128+5184k^2)\beta^2\bigr).
\end{multline}
Our goal is to prove that, for every $k\geq 0$, we can find $\beta$ such that $H_{\beta}(x,k)\geq 0$ for all values of $x\geq x_{+}(k)$. Since the multiplicative factor in \eqref{Hbeta} is always positive after the event horizon, we can equivalently prove positivity of $h_{\beta}(x,k)$. After numerical inspection, we distinguish two cases.

\subsubsection{Case 1: \texorpdfstring{$k\in[0,4]$}{k less than 4}}

Numerical inspection suggest $\beta=1$ for $k\in[0,4/9)$, $\beta=2$ for $k\in[4/9,18/10]$, $\beta=3$ for $k\in[18/10,3]$, and $\beta=4$ for $k\in[3,4]$. In all four cases, positivity can be proven as follows. Express $h_\beta(x,k)=A(x)+B(x)k^2+C(x)k^4$ and consider the parameter $u=k^2$. Verify that $h_\beta$ is a convex parabola w.r.t $u$, locate its minimum in the relative $u$ interval, which could be one of the endpoints of the interval or the vertex of the parabola, and prove positivity at the minimum.
\medskip

We complete a proof in full detail for $\beta=1$ and $k\in[0,4/9)$\footnote{Note that, in the asymptotically flat limit $l\rightarrow\infty$, the same choice for $g(r)$ is made in \cite{holzegel2024note} for Schwarzschild.}; a completely analogous proof can be completed for each of the other sub-intervals of $k$, with their respective $\beta$. 

Begin by writing $h_1(x,k) \;=\; A(x)+B(x)\,k^2+C(x)\,k^4$ with
\begin{align*}
    A(x) &= 96x^7+52x^6-96x^5+945x^4-864x^3-4482x^2+5184x-243, \\
    B(x) &= 8x^3\!\left(16x^6-96x^5+162x^4+1269x^3+432x^2-2754x+243\right), \\
    C(x) &= 432x^6(8x-3)(2x(x+2)+3).
\end{align*}

\medskip
\noindent\textbf{Step 0 (Event horizon location).} 
Since $x_+(k)\geq x_+(4/9)\approx 0.71$ for all $k\in[0,4/9)$, it suffices to verify positivity for $x\geq 7/10$.

\medskip
\noindent\textbf{Step 1 ($h_1$ is convex in $u$).} 
Since $C(x)>0$ for all $x>3/8$, and $7/10>3/8$, the function $h_1$ is an upward-opening 
parabola in $u$ for all $x\geq 7/10$.

\medskip
\noindent\textbf{Step 2 (Location of the minimum).} 
The vertex is at $u^*(x)=-B(x)/(2C(x))$. To determine whether $u^*(x)$ lies inside 
$(0,16/81)$, we evaluate the derivative at the right endpoint:
\[ \frac{\partial h_1}{\partial u}\bigg|_{u=\frac{16}{81}}
    = B(x)+\frac{32}{81}C(x) = \frac{8}{3}x^3\,q(x),\]
with $q(x) = 1072x^6+1376x^5+1254x^4+3231x^3+1296x^2-8262x+729$.
We analyse the signs of $q(x)$ and $B(x)$.
\begin{description}
    \item[Sign of $q(x)$:] Since $q(7/10)<0$ and $q(1)>0$, there is at least one root in 
    $(7/10,1)$. Descartes' Rule gives $q$ at most $2$ positive roots, so there is exactly 
    one such root $x_1\in(7/10,1)$, and $q(x)>0$ for all $x>x_1$. Explicit evaluation 
    gives $x_1\in(96/100,97/100)$.

    \item[Sign of $B(x)$:] Since $B(1)<0$ and $B(2)>0$, there is at least one root 
    $x_2\in(1,2)$. Sturm's Theorem (see \cite{lang2012algebra}, Chapter XI, Section 2, Theorem 2.7) confirms exactly one root in $(7/10,+\infty)$, so 
    $B(x)<0$ on $(7/10,x_2)$ and $B(x)\geq 0$ on $[x_2,+\infty)$.
\end{description}

This yields three regions where the minimum location differs:

\begin{center}
\begin{tabular}{c|c|c|c}
    \textbf{Region} & \textbf{Range of $x$} & \textbf{Sign of $q$, $B$} 
        & \textbf{Minimum at} \\ \hline
    I   & $(7/10,\,x_1]$      & $q\leq 0$, $B$ arbitrary & $u = 16/81$ \\
    II  & $(x_1,\,x_2)$       & $q>0$, $B<0$             & $u = u^*(x)$ (vertex) \\
    III & $[x_2,\,+\infty)$   & $q>0$, $B\geq 0$         & $u = 0$ \\
\end{tabular}
\end{center}

\medskip
\noindent\textbf{Step 3 (Positivity at the minimum).}

\begin{description}
    \item[Region I.] We need $h_1(x,4/9)>0$ for $x\geq 7/10$. By Sturm's Theorem, 
    $h_1(x,4/9)$ has no real roots in $[7/10,+\infty)$, and since 
    $h_1(7/10,4/9)>0$, we conclude $h_1(x,4/9)>0$ throughout.

    \item[Region II.] The value of $h_1$ at its vertex is
    \[
        h_1(x,\sqrt{u^*(x)}) \;=\; A(x)-\frac{(B(x))^2}{4C(x)} 
        \;=\; -\frac{x^2(x+3)^2(2x-3)^2}{27(8x-3)(2x^2+4x+3)}\,\eta(x),
    \]
    where $\eta(x)=64x^6-960x^5-3456x^4-15336x^3-49005x^2-7290x+64152$. We need 
    $\eta(x)<0$ on $(x_1,x_2)$. By Sturm's Theorem, $\eta$ has exactly two positive 
    roots $x_0<x_0'$. Since $\eta(7/10)>0$ and $\eta(95/100)<0$, we have 
    $x_0\in(7/10,95/100)$. As $x_1\in(96/100,97/100)$, we get $x_0<x_1$, so $\eta<0$ on 
    $(x_1,x_0')$. Since $\eta(2)<0$ and $x_2\in(1,2)$, we have $x_0'>2>x_2$. Since the top order coefficient is positive and $\eta$ has exactly two positive roots, it must be negative between them, that is $\eta<0$ throughout $(x_1,x_2)$.

    \item[Region III.] We need $A(x)>0$ for $x\geq x_2$. In fact the stronger statement 
    $A(x)>0$ for all $x\geq 7/10$ holds: by Sturm's Theorem $A$ has no real roots in 
    $[7/10,+\infty)$, and direct evaluation yields $A(7/10)>0$.
\end{description}

\subsubsection{Case 2: \texorpdfstring{$k\geq 4$}{k bigger than 4}}

Numerical inspection suggests $\beta=2k$, in which case we have
\begin{multline}
    h_{2k}(x,k) \;=\; -243+(-162+10692k)x+(-27-1296-15288k^2)x^2 + \\ 
    (-4860k+8208k^2)x^3+(-46656 k^3+7236 k^2-1080 k) x^4+ (57024 k^4-22032 k^3-168 k^2) x^5+ \\
    (30672 k^4+3024 k^3+208 k^2) x^6+ (384 k^2 + 1728 k^3 - 864 k^4 + 15552 k^5) x^7+\\
    (-3504 k^4 + 15552 k^5 + 15552 k^6) x^8 + (512 k^4 + 3456 k^5 + 
    20736 k^6) x^9.
\end{multline}
We rely on the following inequality, which is valid for $k\geq 1$:
\[ \Big(\frac{1}{8k^2}\Big)^{\frac13}\leq x_+(k) \leq \Big(\frac{1}{4k^2}\Big)^{\frac13}.\]
Using the above, we make the following monotonicity argument:
\begin{enumerate}
    \item Prove that $\de_x^3h_{2k}(x,k)>0$ for all $k\geq 1$ and all $x\geq \Big(\frac{1}{8k^2}\Big)^{\frac13}$;
    \item Since $\de_x^2h_{2k}((\frac{1}{8k^2})^{\frac13} ,k)>0$, (1) implies $\de_x^2h_{2k}(x,k)>0$ for all $k\geq 1$ and all $x\geq \Big(\frac{1}{8k^2}\Big)^{\frac13}$;
    \item Since $\de_xh_{2k}((\frac{1}{8k^2})^{\frac13},k)>0$, (2) implies $\de_xh_{2k}(x,k)>0$ for all $k\geq 1$ and all $x\geq \Big(\frac{1}{8k^2}\Big)^{\frac13}$;
    \item Since $h_{2k}((\frac{1}{8k^2})^{\frac13},k)>0$, (3) implies $h_{2k}(x,k)>0$ for all $k\geq 4$ and all $x\geq \Big(\frac{1}{8k^2}\Big)^{\frac13}$;
\end{enumerate}
\textbf{Step 1: positivity of the third derivative.} Compute
\begin{multline}
\frac{1}{24k}\de_x^3h_{2k}(x,k)= -1215 + 2052 k + (-1080 + 7236 k - 46656 k^2) x \\
+(-420 k - 55080 k^2 + 142560 k^3) x^2 + (1040 k + 15120 k^2 + 
    153360 k^3) x^3 \\
    +(3360 k + 15120 k^2 - 7560 k^3 + 136080 k^4) x^4 + (-49056 k^3 + 
    217728 k^4 + 217728 k^5) x^5 \\
    (10752 k^3 + 72576 k^4 + 435456 k^5) x^6.
\end{multline}
and observe that the only negative coefficient, in the range $k\geq 1$, is the first order one. We evaluate the second degree polynomial
\[ p_2(x,k)=-1215 + 2052 k + (-1080 + 7236 k - 46656 k^2) x+(-420 k - 55080 k^2 + 142560 k^3) x^2\]
at $x=(\frac{1}{8k^2})^{\frac13}$ and set $k=t^3$ to find
\[\dfrac{-540 - 105 t - 1215 t^2 + 3618 t^3 - 13770 t^4 + 2052 t^5 - 
 23328 t^6 + 35640 t^7}{t^2}. \]
By Sturm's Theorem, the numerator has no roots between $1$ and infinity and is therefore positive. This means that we are in a region where the upward pointing parabola $p_2(x,k)$ is positive; to ensure that we are in the region \emph{after} its zeroes, we check that its first derivative is positive at $(1/8k^2)^{\frac13}$. Compute $\de_x p_2((1/8k^2)^{\frac13},k)$, set $k=t^3$ and obtain
\[ 12 (-90 - 35 t + 9 t^3 (67 + 6 t (-85 + 4 t^2 (-18 + 55 t)))), \]
which has no roots between 1 and infinity by Sturm's Theorem and is thus positive.

\medskip

\noindent \textbf{Step 2: positivity of the second derivative.} Compute $\de_x^2h_{2k}((1/8k^2)^{\frac13},k)$, set $k=t^3$ and obtain
\[  \dfrac{6\Big(84 + 113 t^2 - 162 t^3 - 590 t^4 - 1161 t^5 + 3429 t^6 - 7344 t^7 + 
 15633 t^8 - 19926 t^9 + 20952 t^{10} \Big)}{t^4}. \]
By Sturm's Theorem the numerator has no roots between $1$ and infinity, thus implying positivity.

\medskip

\noindent \textbf{Step 3: the first derivative.}
Evaluating $\de_xh_{2k}(x,k)$ at $x=(1/8k^2)^{\frac13}$ and setting $k=t^3$ we find
\[ \dfrac{28 + 38 t^2 - 234 t^3 - 199 t^4 - 1971 t^5 + 2241 t^6 - 4806 t^7 + 
 8424 t^8 - 7290 t^9 + 2376 t^{10}}{4t^6}, \]
which, by Sturm's theorem, has no roots between $t=\sqrt[3]{4}$ and infinity and is thus positive.

\medskip

\noindent \textbf{Step 4: the function.} The function $h_{2k}$ is not positive at the lower bound $x=(1/8k^2)^{\frac13}$. However, we can express $k^2$ as a function of $x_+$ as $k^2=\frac{1-x_+}{2x_+^3}$ using that $x_+ - 1 + 4 k^2 x_+^3=0$, and obtain
\[ h_{2k(x_+)}(x_+,k(x_+))=(3+x_+)^2h_+(x_+) \]
with $h_+(x)=A(x)+\sqrt{x(1-x)}B(x)$, $A(x)=145x^2-64x^3-108x$, $B(x)=216-108x$. For $k\geq 4$, the event horizon $x_+(k)\in[0,23/100)$, so we need positivity of $h_+(x)$ for values of $x$ between $0$ and $23/100$. Since $A(x)\leq 0$ and $B(x)>0$, we obtain positivity by showing that  
\[ x(1-x)B^2(x)-A^2(x)\geq 0 \quad \text{for} \quad x\in[0,23/100],\]
which is true by direct evaluation and Sturm's Theorem.

In conclusion, since $h_{2k}(x,k)$ is always positive at the event horizon $x_+(k)$ and monotonically increasing for all values of $x\geq x_+(k)$, for all $k\geq 4$, it is positive for all $x\geq x_+(k)$. This proves positivity of the zeroth order coefficient $H_\beta$ and achieves the goal of this Appendix.
\newpage

\printbibliography
\end{document}